\documentclass[11pt]{article}
\pdfoutput=1
\usepackage[T1]{fontenc}
\usepackage[utf8]{inputenc}

\usepackage[final]{microtype}

\usepackage{bbm}

\usepackage[%
	backend=biber,
	style=alphabetic,
	sorting=nyt,
	url=false,
	isbn=true,
	maxbibnames=10,
]{biblatex}
\usepackage{authblk}

\DeclareSourcemap{%
	\maps[datatype=bibtex]{%
		\map[overwrite]{%
			\step[fieldsource=doi, final]
			\step[fieldset=url, null]
			\step[fieldset=eprint, null]
		}
	}
}

\usepackage{graphicx}
\usepackage{svg}

\usepackage{amssymb,amsmath,amsthm}

\usepackage{bm}

\usepackage{mathtools}
\mathtoolsset{centercolon}
\DeclareMathOperator{\polyOp}{poly}

\DeclareMathOperator{\rt}{rt}
\DeclareMathOperator{\fastrt}{rt_f}
\DeclareMathOperator{\qrt}{qrt}
\DeclareMathOperator{\fastqrt}{qrt_f}
\DeclareMathOperator{\cqrt}{cqrt}

\DeclareMathOperator{\hqrt}{hqrt}
\DeclareMathOperator{\hqrtA}{hqrt_a}
\DeclareMathOperator{\fasthqrtA}{hqrt_{fa}}
\DeclareMathOperator{\rs}{rs}
\DeclareMathOperator{\qrs}{qrs}
\DeclareMathOperator{\twoqubitgate}{tg}
\DeclareMathOperator{\diam}{diam}

\DeclareMathOperator{\Ex}{\mathbf{E}}
\DeclareMathOperator{\Prob}{P}
\DeclareMathOperator{\tr}{Tr}
\DeclareMathOperator{\entanglementOp}{E}

\DeclarePairedDelimiter\ceil{\lceil}{\rceil}

\DeclarePairedDelimiter\bra{\langle}{\vert}
\DeclarePairedDelimiter\ket{\vert}{\rangle}
\DeclarePairedDelimiter\set{\{}{\}}
\DeclarePairedDelimiter\abs{\lvert}{\rvert}
\DeclarePairedDelimiter{\norm}{\lVert}{\rVert}
\DeclarePairedDelimiterX{\inp}[2]{\langle}{\rangle}{#1, #2} % Inner product
\DeclarePairedDelimiterXPP\tg[1]{\twoqubitgate}{(}{)}{}{#1}
\DeclarePairedDelimiterXPP\bigo[1]{O}{(}{)}{}{#1}
\DeclarePairedDelimiterXPP\littleo[1]{o}{(}{)}{}{#1}
\DeclarePairedDelimiterXPP\bigomega[1]{\Omega}{(}{)}{}{#1}
\DeclarePairedDelimiterXPP\bigtheta[1]{\Theta}{(}{)}{}{#1}
\DeclarePairedDelimiterXPP\poly[1]{\polyOp}{(}{)}{}{#1}
\DeclarePairedDelimiterXPP\rnumber[1]{\rt}{(}{)}{}{#1}
\DeclarePairedDelimiterXPP\rfnumber[1]{\fastrt}{(}{)}{}{#1}
\DeclarePairedDelimiterXPP\qrnumber[1]{\qrt}{(}{)}{}{#1}
\DeclarePairedDelimiterXPP\qrfnumber[1]{\fastqrt}{(}{)}{}{#1}
\DeclarePairedDelimiterXPP\cqrnumber[1]{\cqrt}{(}{)}{}{#1}
\DeclarePairedDelimiterXPP\hrnumber[1]{\hqrt}{(}{)}{}{#1}
\DeclarePairedDelimiterXPP\hrAnumber[1]{\hqrtA}{(}{)}{}{#1}
\DeclarePairedDelimiterXPP\hrfanumber[1]{\fasthqrtA}{(}{)}{}{#1}
\DeclarePairedDelimiterXPP\rsize[1]{\rs}{(}{)}{}{#1}
\DeclarePairedDelimiterXPP\qrsize[1]{\qrs}{(}{)}{}{#1}
\DeclarePairedDelimiterXPP\probability[1]{\Prob}{[}{]}{}{#1}
\DeclarePairedDelimiterXPP\expectation[1]{\Ex}{[}{]}{}{#1}
\DeclarePairedDelimiterXPP\trace[1]{\tr}{(}{)}{}{#1}
\DeclarePairedDelimiterXPP\ptrace[2]{\tr_{#1}}{(}{)}{}{#2}
\DeclarePairedDelimiterXPP\entanglement[1]{\entanglementOp}{(}{)}{}{#1}
\DeclarePairedDelimiterXPP\capacity[2]{\Gamma}{(}{)}{}{#1,#2}
\DeclarePairedDelimiterXPP\matching[1]{m}{(}{)}{}{#1}

\makeatletter
\def\Ddots{\mathinner{\mkern1mu\raise\p@
\vbox{\kern7\p@\hbox{.}}\mkern2mu
\raise4\p@\hbox{.}\mkern2mu\raise7\p@\hbox{.}\mkern1mu}}
\makeatother

\newcommand*{\tran}{^{\mkern-1.5mu\mathsf{T}}} % transpose
\newcommand*{\idm}{\ensuremath{\mathbbm{1}}}
\renewcommand{\vec}[1]{\ensuremath{\bm{#1}}}

\usepackage{booktabs}

\usepackage[lined,algoruled,algosection,linesnumbered,noend,hanginginout]{algorithm2e}
\SetAlCapSkip{0.5em}
\SetEndCharOfAlgoLine{} % Remove semicolons at end of line
\SetArgSty{textnormal}
\SetKwInput{Data}{Data}
\SetKwInput{Input}{Input}
\SetKwInput{Output}{Output}
\SetKwFor{For}{for}{\string:}{endfor}
\SetKwFor{While}{while}{\string:}{endw}
\SetKwFor{ForEach}{foreach}{\string:}{endfch}
\SetKwProg{Fn}{function}{\string:}{endfunction}
\usepackage[colorinlistoftodos]{todonotes}
\usepackage[font=small,labelfont=bf]{caption}
\usepackage{subcaption}

\usepackage{color}
\usepackage{xcolor}
\definecolor{dark-red}{rgb}{0.4,0.15,0.15}
\definecolor{dark-blue}{rgb}{0.15,0.15,0.4}
\definecolor{medium-blue}{rgb}{0,0,0.5}

\definecolor{mycomment}{rgb}{0.3,0.7,0.8}
\definecolor{mygray}{rgb}{0.5,0.5,0.5}
\definecolor{lightgray}{rgb}{0.95,0.95,0.95}
\definecolor{mymauve}{rgb}{0.58,0,0.82}

\usepackage[unicode=true]{hyperref}
\hypersetup{%
	breaklinks=true,
	colorlinks,
	linkcolor={dark-blue},
	citecolor={dark-blue},
	urlcolor={medium-blue},
	pdfpagemode=UseOutlines,
	pdfborder={0 0 0}
}
\usepackage{enumitem}
\newlist{ienumerate}{enumerate*}{1}
\setlist*[ienumerate,1]{%
	label=(\roman*),
}

\usepackage[nameinlink,capitalise]{cleveref}
\crefname{figure}{Figure}{Figures}
\crefname{equation}{}{} % Remove Eq. from equation reference
\Crefname{equation}{Eq.}{Eqs.} % Except at the start of a sentence (for \Cref)

\newtheorem{theorem}{Theorem}[section]
\newtheorem{lemma}[theorem]{Lemma}
\newtheorem{corollary}[theorem]{Corollary}

\theoremstyle{definition}
\newtheorem{definition}[theorem]{Definition}

\newcommand{\cnot}[0]{{\textsc{cnot}}}
\newcommand{\cz}[0]{{\textsc{cz}}}

\newcommand{\cknot}[1]{\ensuremath{\textsc{c}^{#1}\kern-0.1em\textsc{not}}}

\newcommand{\swap}[0]{{\textsc{swap}}}
\newcommand{\RoutingViaMatchings}[0]{{\textsc{Routing via Matchings}}}

\newcommand{\sieConst}{\ensuremath{\alpha}}

\newcommand{\paren}[1]{\left(#1\right)}

\author[1,2,3,4]{Aniruddha Bapat\thanks{\href{mailto:ani@lbl.gov}{ani@lbl.gov}}}
\author[2,5,6]{Andrew~M.~Childs\thanks{\href{mailto:amchilds@umd.edu}{amchilds@umd.edu}}}
\author[2,3]{Alexey~V.~Gorshkov\thanks{\href{mailto:gorshkov@umd.edu}{gorshkov@umd.edu}}}
\author[2,5,6,7]{Eddie Schoute\thanks{\href{mailto:eschoute@lanl.gov}{eschoute@lanl.gov}}}
\affil[1]{Lawrence Berkeley National Laboratory, Berkeley, CA 94720, USA}
\affil[2]{Joint Center for Quantum Information and Computer Science, NIST/University of Maryland, College Park, MD 20742, USA}
\affil[3]{Joint Quantum Institute, NIST/University of Maryland, College Park, MD 20742, USA}
\affil[4]{Department of Physics, University of Maryland, College Park, MD 20742, USA}
\affil[5]{Institute for Advanced Computer Studies, University of Maryland, College Park, MD 20742, USA}
\affil[6]{Department of Computer Science, University of Maryland, College Park, MD 20742, USA}
\affil[7]{Computer, Computational, and Statistical Sciences Division, Los Alamos National Laboratory, Los Alamos, NM 87545, USA}
\title{Advantages and limitations of quantum routing}

\begin{document}

\maketitle

\begin{abstract}%
  \noindent
  The \textsc{Swap} gate is a ubiquitous tool for moving information on quantum hardware,
  yet it can be considered a classical operation because it does not entangle product states.
  Genuinely quantum operations could outperform \swap{} for the task of permuting qubits within an architecture, which we call \emph{routing}.
  We consider quantum routing in two models:
  \begin{ienumerate}
  \item allowing arbitrary two-qubit unitaries, or
  \item allowing Hamiltonians with norm-bounded interactions.
  \end{ienumerate}
  We lower bound the circuit depth or time of quantum routing in terms
  of spectral properties of graphs representing the architecture interaction constraints,
  and give a generalized upper bound for all simple connected $n$-vertex graphs.
  In particular, we give conditions for a superpolynomial classical-quantum routing separation,
  which exclude graphs with a small spectral gap and graphs of bounded degree.
  Finally, we provide examples of a quadratic
  separation between gate-based and Hamiltonian routing models
  with a constant number of local ancillas per qubit
  and of an $\bigomega{n}$ speedup if we also allow fast local interactions.
\end{abstract}

\section{Introduction}
% Motivation why routing is relevant
Scalable quantum architectures are expected to have geometrically constrained interactions~\cite{Kielpinski2002,Murali2020,Monroe2013,Monroe2014,Brecht2016,Jones2012}.
Unlike conversions between gate sets, which introduce only a logarithmic overhead due to the Solovay-Kitaev theorem~\cite{Kitaev1997},
architecture connectivity can introduce a polynomial overhead from the cost of simulating nonlocal interactions.
For example, a unitary implementation of a \cnot{} gate on the ends of an $n$-qubit 1D chain requires time $\bigomega{n}$
by a signaling argument.
This raises a natural question: how do we implement general nonlocal operations in minimal depth under architectural constraints?

% What is routing and connections to QI
A natural approach to implementing nonlocal gates is to first permute qubits within the architecture.
We call the task of implementing an arbitrary given permutation of qubits via operations on neighboring qubits \emph{routing}.
Routing generalizes well-studied tasks such as
state transfer~\cite{Bose2003,Bose2007,Christandl2005,Venuti2007,Franco2008,Banchi2011,Yao2011}
and state reversal (or mirroring)~\cite{Albanese2004,Shi2005,Karbach2005,Raussendorf2005,Fitzsimons2006}.
By exploring limits on routing, we also explore limits on information transfer and entanglement generation.
In particular, the entanglement across a bipartition can be increased by routing,
e.g., by starting with a maximally entangled qubit pair in the left partition and sending one qubit across the bipartition.
Therefore, bounds on routing relate to bounds on entanglement capacity~\cite{Duer2001,Childs2003,Childs2004,Bennett2003}
and have further relations to Lieb-Robinson bounds~\cite{Lieb1972}.

% Classical routing and background
A common implementation of routing uses \swap{} gates to implement permutations~\cite{Saeedi2011,Lye2015,Zulehner2019,Childs2019}.
We call this approach \emph{classical routing} since a separable state acted on by \swap{} gates
cannot become entangled.
In routing, we represent architecture connectivity by a simple connected graph, $G$.
Classical routing algorithms have been developed using a variety of techniques,
including shortest path algorithms~\cite{Metodi2006,Saeedi2011,Lin2015,Wille2016,Murali2019,Qiskit},
sorting algorithms~\cite{Hirata2009,Beals2013,Shafaei2014,Pedram2016,Brierley2017},
routing on a spanning tree~\cite{Maslov2008},
and exhaustive search~\cite{Lye2015,Zulehner2019}.
In fact, classical routing is equivalent to the \RoutingViaMatchings{} problem in classical computer science~\cite{Childs2019}.
\RoutingViaMatchings{} is NP-complete~\cite{Banerjee2017,Miltzow2016},
but efficient algorithms exist for special cases of architecture connectivity such as paths, complete graphs, trees,
and graph products that capture practical architectures such as grids~\cite{Alon1994,Zhang1999,Childs2019}.
A natural lower bound on classical routing arises from small vertex cuts in the architecture~\cite{Alon1994}.

\begin{figure}
	\centering
	\includeinkscape[width=0.8\textwidth]{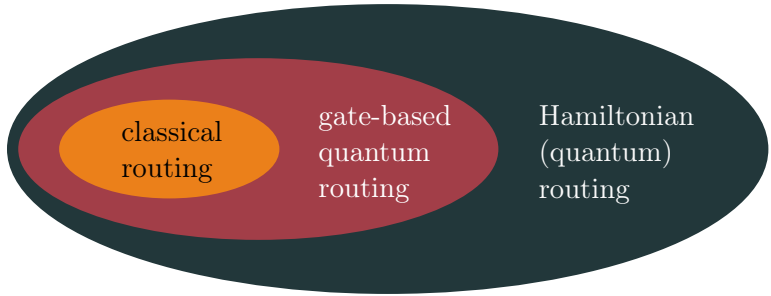}
	\caption{%
		The relative power of routing models considered in this work.
		Most prior work only considers using \swap{} gates for routing, which we call \emph{classical routing.}
		In this work, we explore the additional routing power provided by genuinely quantum operations.
		We consider increasingly more powerful quantum routing models:
		first, allowing arbitrary two-qubit gates in \emph{gate-based quantum routing},
		and then allowing continuous Hamiltonian evolution in \emph{Hamiltonian (quantum) routing}.
	}\label{fig:routing-types}
\end{figure}

% Contribution
In this work, we explore the extent to which genuinely quantum operations can accelerate routing
in what we call \emph{quantum routing}.
The relative power of the models we explore is depicted in \cref{fig:routing-types}.
In \cref{sec:quantumRouting}, we introduce gate-based quantum routing with arbitrary two-qubit unitaries
and Hamiltonian (quantum) routing by Hamiltonian evolution with norm-bounded interactions.

A key observation is that routing can distribute entanglement across a bipartition of the system.
Thus, by lower bounding the circuit depth to create entanglement across vertex cuts in the architecture,
we derive bounds on quantum routing and improve bounds on classical routing.
% we show an example where the maximum matching lower bound is tighter (asymptotically).
The same argument does not apply in the continuous-time setting of the Hamiltonian model.
However, we show a similar (but weaker) lower bound from
a lower bound on the time to create entanglement across small edge cuts in the architecture
that improves on a previous bound~\cite{Bravyi2006} by constant factors.

We show lower bounds on routing by proving lower bounds on state preparation in the respective models.
The circuit depth to distribute entanglement in the gate-based model is lower bounded by
the maximum matching size in the edge boundary of a vertex cut.
In the Hamiltonian model, the evolution time is lower bounded in terms of edge cuts.
Our state preparation lower bounds generalize earlier analyses for lattices~\cite{Acoleyen2013,Gong2017,Piroli2021} to general graphs.

Next, in \cref{sec:classical}, we investigate separations between Hamiltonian and classical routing.
We prove a general upper bound on classical routing on simple connected graphs,
allowing us to prove that, under certain conditions on the spectral gap of the Laplacian of $G$ and the degrees of its vertices,
there is no superpolynomial separation between the (worst-case) classical and Hamiltonian routing times.
In particular, our results rule out a superpolynomial separation for interaction graphs of bounded (constant) degree,
a common feature of practical quantum architectures~\cite{Nation2021,Arute2019,Kielpinski2002}.

However, we are not aware of even a superconstant speedup of Hamiltonian routing over classical routing for any family of graphs.
In \cref{sec:separation}, we give two such examples strengthened routing models.
The first is an $\bigomega{\sqrt{n}}$ factor speedup in a strengthened routing model with one local ancilla per qubit.
The second, a $\bigtheta{n}$ speedup,
follows from allowing fast local interactions,
which give an asymptotically optimal gate-based routing algorithm
and an asymptotically optimal Hamiltonian routing algorithm if allowed one local ancilla per qubit.

\section{Quantum routing}\label{sec:quantumRouting}
In this section, we introduce quantum routing and prove lower bounds dependent on graph expansion properties.
We model the architectural constraints by a simple graph $G$ on $n$ vertices,
with the qubits represented by the vertex set $V(G)$ and the allowed interactions between qubits by the edge set $E(G)$.

\subsection{Gate-based quantum routing}
First, we consider routing in the gate-based model of quantum computation.
Analogous to the (classical) routing number~\cite{Alon1994} (see also \cref{eq:rnumber} in \cref{sec:classical}),
we define the \emph{gate-based quantum routing number} $\qrnumber{G}$ as
\begin{equation}
	\qrnumber{G} \coloneqq \max_\pi \qrnumber{G, \pi},
\end{equation}
where $\pi$ is a permutation of the qubits and $\qrnumber{G, \pi}$
is the minimum depth of a unitary circuit that implements the permutation $\pi$ while respecting the architecture constraints $G$,
i.e., only having two-qubit gates\footnote{We do not limit routing circuits to a particular gate set.
If necessary, any such circuit can be approximated by a universal, inverse-closed gate set with at most polylogarithmic overhead by the Solovay-Kitaev theorem~\cite{Kitaev1997}.} acting along the edges $E(G)$.
In this model, single-qubit gates are free since they can be absorbed into adjacent two-qubit gates.

We briefly prove a diameter lower bound on gate-based quantum routing.
The diameter of a graph $G$ is 
\begin{equation}
	\diam(G) \coloneqq \max_{u,v \in V(G)} d(u,v),
\end{equation}
where $d(u,v)$ is the (shortest) distance between vertices $u$ and $v$.
\begin{theorem}\label{thm:diamLB}
	For any simple graph $G$,
	\begin{equation}
		\qrnumber{G} \ge \diam(G).
	\end{equation}
\end{theorem}
\begin{proof}
Consider two vertices $u,v \in V(G)$ at a distance $\diam(G)$
and a circuit $\mathcal C$ of two-qubit unitaries with depth $D$ acting on $G$.
Any local operator acting on $u$ %(or $v$) 
evolved in the Heisenberg picture under $\mathcal C$ 
will have no support 
on vertices further than distance $D$.
In order to swap $u$ and $v$, all of the support of that Heisenberg-evolved operator must be on $v$, which implies $D \ge \diam(G)$. Therefore, $\qrnumber{G} \ge \diam(G)$.
\end{proof}

To prove a lower bound on gate-based quantum routing,
we relate routing to the task of generating entanglement.
We can quantify the entanglement of a pure state $\rho$
on a bipartite joint system $X\bar X$,
consisting of the subsystems $X$ and $\bar X$,
by the von Neumann entropy of the reduced density operator $\rho_X \coloneqq \ptrace{\bar X}{\rho}$,
defined as
\begin{equation}
	S(\rho_X) \coloneqq - \trace{\rho_X \log \rho_X}.
\end{equation}
(The function $\log(x)$ denotes the logarithm base 2 unless specified otherwise. We denote the natural logarithm by $\ln(x)$.)
We refer to the von Neumann entropy as ``the entropy''
and denote $S_X(\rho) \coloneqq S(\rho_X)$.
For completeness, we list some elementary properties of the entropy that will be useful later
and can be easily verified.
\begin{lemma}\label{lem:entropyProperties}
	For a state $\rho$ on a joint system $X\bar X$, the following statements about the entropy hold:
	\begin{enumerate}
		\item\label{item:entropySymmetric} If $\rho$ is a pure state,
			then the entropy is symmetric,
			i.e.,
			\begin{equation}
				S_X(\rho) = S_{\bar X}(\rho).
			\end{equation}
		\item\label{item:entropyChangeOfBasis} The entropy is invariant under change of basis, i.e.,
			\begin{equation}
				S(U \rho U^\dagger) = S(\rho).
			\end{equation}
		\item\label{item:entropyLocalX} The entropy is invariant under local unitaries $U_X$ on $X$
			and $U_{\bar X}$ on $\bar X$, i.e.,
			\begin{equation}
				S_X((U_X \otimes U_{\bar X}) \rho (U_X \otimes U_{\bar X})^\dagger) = S_X(\rho).
			\end{equation}
	\end{enumerate}
\end{lemma}

\begin{figure}[tbp]
	\centering
	\includeinkscape{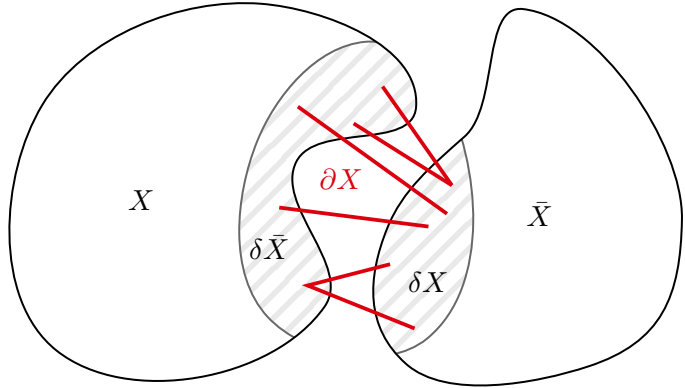}
	\caption{
		A graph can be partitioned into two sets of vertices $X$ and $\bar X$.
		The vertex boundary $\delta X$ of $X$ is the set of vertices outside of $X$
		that are directly connected to $X$,
		and similarly for $\delta \bar X$ and $\bar X$.
		The edge boundary $\partial X$ (red) of $X$ (and $\bar{X}$) is the set of edges that connect $X$ to $\bar{X}$.
	}\label{fig:biparted-graph}
\end{figure}

The change in entropy of the reduced state on $X \subseteq V(G)$
by a unitary respecting the constraints of the interaction graph $G$
can be bounded by a quantity proportional to the vertex boundary of $X$.
The set $X$ and its \emph{vertex complement} $\bar X \coloneqq V(G) \setminus X$ define a vertex partition of $G$;
see \cref{fig:biparted-graph}.
By the invariance of the entropy under local unitaries (part 3 of \cref{lem:entropyProperties}), we need only consider the unitary acting across the partition.
In particular, the unitary must act  on the \emph{vertex boundary}
\begin{equation}
	\delta X \coloneqq \set{v \in \bar X \mid \set{u,v} \in E(G),\, u \in X},
\end{equation}
which forms a vertex cut in $G$.
We formalize this bound on the change in entropy in the following lemma derived from the small total entangling property~\cite{Marien2016}.

\begin{lemma}[Small total entangling (STE)]\label{lem:ste}
	For a unitary $U$ acting nontrivially only on
	the joint subsystem $\delta X \delta \bar X$,
	the change in the entropy of any state $\rho$ is bounded by
	\begin{equation}
		\abs*{S_X(U\rho U^\dagger) - S_X(\rho)} \leq 2 \min(\abs{\delta X}, \abs{\delta \bar X}).
	\end{equation}
\end{lemma}
\begin{proof}
	We consider a purification system $R$ such that $\ptrace{R}{\ket{\psi}\bra{\psi}} = \rho$
	for some pure state $\ket\psi$ on the joint system $X\bar X R$.
	For any subsystem $Y$ of $X\bar X$,
	\begin{equation}
	    S_Y(\ket\psi) \coloneqq S_Y(\ptrace{R}{\ket\psi\bra\psi}).
	\end{equation}
    
	The unitary $U$ is a local unitary on the joint subsystem $\delta X\delta \bar X$, so
	\cref{lem:entropyProperties} implies
	\begin{equation}
		\abs{S_X(U\rho U^\dagger) - S_X(\rho)} = \abs{S_X(U \ket\psi) - S_X(\ket\psi)} = \abs{S_{X|\delta X}(U\ket{\psi}) - S_{X|\delta X}(\ket\psi)}
	\end{equation}
	where we used the conditional quantum entropy, $S_{X|\delta X}(\ket \psi) = S_{X\delta X}(\ket\psi) - S_X(\ket\psi)$.
	By the triangle inequality and \cite[theorem~11.5.1]{WildeBook},
	\begin{equation}
    	\abs{S_{X|\delta X}(U\ket{\psi}) - S_{X|\delta X}(\ket\psi)} \le \abs{S_{X|\delta X}(U\ket{\psi})} + \abs{S_{X|\delta X}(\ket\psi)}
				   		   \le 2 \abs{\delta X}.\label{eq:ste1}
	\end{equation}
	By symmetry of the entropy for pure states, we also obtain
	\begin{equation}
		\abs{S_X(U\rho U^\dagger) - S_X(\rho)} = \abs{S_{\bar X | \delta \bar X}(U\ket{\psi}) - S_{\bar X | \delta \bar X}(\ket\psi)}
				   		   \le 2 \abs{\delta \bar X}.\label{eq:ste2}
	\end{equation}
	The minimum of \cref{eq:ste1,eq:ste2} gives the required bound.
\end{proof}

We can saturate this bound in several special cases.
A \swap{} gate can saturate this bound when the subsystems $\delta X$ and $\delta \bar X$ are single qubits
that are maximally entangled with the remainder of $\bar X$ and $X$, respectively.
Furthermore, with sufficient connectivity,
we can also saturate this bound in higher dimensions:
let
\begin{equation}
	\abs{\delta X} = \abs{\delta \bar X} \le \min(\abs{X}, \abs{\bar X})/2
\end{equation}
and let $\delta X$ and $\delta \bar X$ be maximally entangled with the remainder of $\bar X$ and $X$, respectively.
Then, if we exchange $\delta X$ with $\delta \bar X$ through simultaneous \swap{}s,
the entropy increases by $2\abs{\delta X}$, saturating the bound.

We now prove a lower bound on the time required for \emph{state preparation} of entangled states
based on the maximum matching size in the edge boundary of the vertex cut.
A \emph{matching} is a set of edges $E' \subseteq E(G)$ such that all vertices in $E'$ are distinct.
For any $E' \subseteq E(G)$, we define $\matching{E'} \subseteq E'$ as the maximum(-size) matching in $E'$.
State preparation is the task of preparing some target state $\rho$ given an initial state $\rho_0$.
A special case of state preparation is routing a particular state.
If the change in entanglement between initial state $\rho_0$ and final state $\rho$ is
$\abs{S_X(\rho) - S_X(\rho_0)}$,
then a simple argument from STE gives a circuit depth lower bound of $\abs{S_X(\rho) - S_X(\rho_0)} / (2\abs{\delta X})$,
and similar arguments have been used with the entanglement capacity~\cite{Bennett2003,Eldredge2020}.
However, this does not account for the time required to entangle the boundary subsystem with the bulk subsystem.
A careful accounting gives the following, which we later show can be saturated.

\begin{lemma}\label{lem:gateBasedStatePreparation}
	Given an initial state $\rho_0$ and a target state $\rho$
	on the bipartite system consisting of $X$ and $\bar X$,
	define the change in entropy 
	\begin{equation}
		\Delta S_Z \coloneqq \abs*{S_Z(\rho) - S_Z(\rho_0)}
	\end{equation}
	for any subsystem $Z$.
	Then any gate-based unitary circuit $\mathcal C$ for preparing $\rho$ from $\rho_0$
	restricted by an interaction graph $G$ has depth
	\begin{equation}\label{eq:gateBasedStatePreparation1}
		d \ge \frac{\Delta S_X}{2\abs{\matching{\partial X}}},
	\end{equation}
	for any $X \subsetneq V(G)$,
	and
	\begin{equation}\label{eq:gateBasedStatePreparation2}
		d \ge \frac{\Delta S_X + \Delta S_Y}{2\abs{\matching{\partial(\delta X)}}},
	\end{equation}
	for $Y \coloneqq \bar X \setminus \delta X$.
\end{lemma}
\begin{proof}

	We can decompose $\mathcal C$ into a sequence of disjoint unitaries $U_i$ acting on $X\delta X$
	and unitaries $V_i$ acting on $Y\delta X$,
	where $i \in \mathbb N$.
	To perform operations $U_i$ and $V_i$ simultaneously, they must act on disjoint subsets
	$X_i, X_i' \subseteq \delta X$, respectively.
	Between each application of $U_iV_i$,
	there are local unitary operations within $X$, $\delta X$, and $Y$,
	labelled as $O_i$,
	that we allow to be performed instantaneously.
	The circuit can thus be decomposed as
	\begin{equation}
		\mathcal C = O_d U_d V_d \dots O_1 U_1 V_1 O_0.
	\end{equation}

	We lower bound $d$ by considering the change in entropy and applying STE (\cref{lem:ste}).
	First, we note that the operations $O_i$ cannot change the entropy of the respective subsystems.
	By STE, $U_i$ can change the entropy of $X$ by at most $2\abs{X_i}$
	and $V_i$ can change the entropy of $Y$ by at most $2\abs{X_i'}$.
	Therefore, we have two inequalities that must be satisfied:
	\begin{align}
		\Delta S_X &\le 2\sum_{i=1}^d \abs{X_i}\\
		\text{and}\quad \Delta S_Y &\le 2\sum_{i=1}^d \abs{X_i'}.
	\end{align}
	Noting that $\abs{X_i} \le \abs{\matching{\partial X}}$,
	we obtain $\Delta S_X \le 2d\abs{\matching{\partial X}}$, thus proving \cref{eq:gateBasedStatePreparation1}.
	Additionally, we note that $\abs{X_i} + \abs{X_i'} \le \abs{\matching{\partial(\delta X)}}$ so that
	\begin{equation}
		\Delta S_X + \Delta S_Y \le 2\sum_{i=1}^d \paren{\abs{X_i} + \abs{X_i'}} \le 2d\abs{\matching{\partial(\delta X)}},
	\end{equation}
	which implies \cref{eq:gateBasedStatePreparation2}.
\end{proof}

Entanglement capacity-based state preparation lower bounds that are proportional to $\abs{\delta X}$~\cite{Gong2017}
can be weaker than \cref{lem:gateBasedStatePreparation} by a factor $\bigomega{n}$ 
for some partitions $X$.
To see this,
consider the graph $L_{2n}$ that consists of two complete graphs $G_1 = K_n$ and $G_2 = K_n$
with additional edges
\begin{equation}
	(\set{x_1} \times V(G_2)) \cup (V(G_1) \times \set{x_2})
\end{equation}
where we pick arbitrary vertices $x_1 \in V(G_1)$ and $x_2 \in V(G_2)$.
For the partition with $X = V(G_1)$,
we have $\abs{\matching{\partial X}} = 2$ whereas $\abs{\delta X} = n$.

Even so, we can obtain a simpler lower bound on the circuit depth as a corollary by relating the change in entropy $\Delta S_{\bar X}$
to that of the bulk system $\bar X \setminus \delta X$.

\begin{corollary}\label{cor:gateBasedVertexLB}
	Given an initial state $\rho_0$ and target state $\rho$,
	the depth of a gate-based state preparation circuit
	restricted by interaction graph $G$ with partition $X \subsetneq V(G)$ is lower bounded by
	\begin{equation}
		d \ge \frac{\Delta S_X + \Delta S_{\bar X} - 2\abs{\delta X}}{2\abs{\matching{\partial(\delta X)}}} \ge \frac{\Delta S_X + \Delta S_{\bar X}}{2\abs{\delta X}} - 1.
	\end{equation}
\end{corollary}
\begin{proof}
	Let $Y \coloneqq \bar X \setminus \delta X$.
	The entropy of the target state can be upper bounded
	using subadditivity and $S_{\delta X}(\sigma) \le \abs{\delta X}$ for any state $\sigma$
	as
	\begin{equation}
		\Delta S_{\bar X} = S_{\bar X}(\rho) - S_{\bar X}(\rho_0) \le S_{Y}(\rho) - S_{\bar X}(\rho_0) + \abs{\delta X}.
	\end{equation}
	The Araki-Lieb triangle inequality,
	\begin{equation}
		S_{\bar X}(\rho_0) \ge \abs{S_{Y}(\rho_0) - S_{\delta X}(\rho_0)},
	\end{equation}
 	then gives
	\begin{equation}\label{eq:complementEntropyIncrease}
		\Delta S_{\bar X} \le S_{Y}(\rho) - S_{\bar X}(\rho_0) + \abs{\delta X} \le \Delta S_Y + 2\abs{\delta X}.
	\end{equation}
	We now apply \cref{lem:gateBasedStatePreparation}, giving
	\begin{equation}
		d \ge \frac{\Delta S_X + \Delta S_{\bar X} - 2\abs{\delta X}}{2\abs{\matching{\partial(\delta X)}}} \ge \frac{\Delta S_X + \Delta S_{\bar X}}{2\abs{\delta X}} - 1
	\end{equation}
	as claimed,
	where the second inequality follows from $\abs{\matching{\partial(\delta X)}} \le \abs{\delta X}$.
\end{proof}

\begin{algorithm}[tbp]
    \caption{%
        Saturating \cref{cor:gateBasedVertexLB} when $\abs{X} = 2k\abs{\delta X} \le \abs{\bar X}$
        for integer $k > 0$,
        assuming every vertex in $\delta X$ is connected to all vertices in $X$
        and $Y \coloneqq \bar X \setminus \delta X$.
        We prepare Bell pairs on $X$ and $\bar X$ and route one end of each Bell pair in $X$ to $\bar X$ and vice versa.
    }\label{alg:saturateStatePreparation}
    \Input{%
        Initial state of $\abs{X}/2$ Bell pairs on each of $X$ and $\bar X$
        where each $v \in \delta X$ is part of a distinct Bell pair.
    }
    simultaneously \swap{} ends of unrouted Bell pairs in $X$ and $\delta X$\;
    \While{$X$ has unrouted ends of Bell pairs}{%
        simultaneously \swap{} ends of unrouted Bell pairs in $Y$ and $\delta X$\;
        simultaneously \swap{} ends of unrouted Bell pairs in $X$ and $\delta X$\;
    }
\end{algorithm}

\cref{cor:gateBasedVertexLB} can be saturated by \cref{alg:saturateStatePreparation}
when $\abs{X} = 2k\abs{\delta X}$ for integer $k > 0$, so that a set of $2\abs{\delta X}$ ends of Bell pairs can be exchanged between $X$ and $\bar X$ every odd time step.
The algorithm makes the additional assumptions that $\abs{X} \le \abs{\bar X}$ to allow for $\Delta S_X =\Delta S_{\bar X} = \abs{X}$
and that $\delta X$ has high connectivity with the rest of the graph so that ends of Bell pairs can easily be routed to and from $\delta X$.
The algorithm saturates \cref{cor:gateBasedVertexLB} after every odd time step up to and including depth $d = 2 k -1$.

\begin{figure}
	\begin{subfigure}{0.5\textwidth}
		\centering
		\includeinkscape[width=\textwidth]{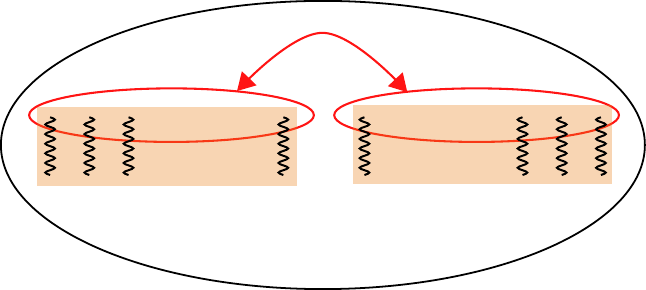}
		\caption{Before routing}
	\end{subfigure}%
	\hfill%
	\begin{subfigure}{0.5\textwidth}
		\centering
		\includeinkscape[width=\textwidth]{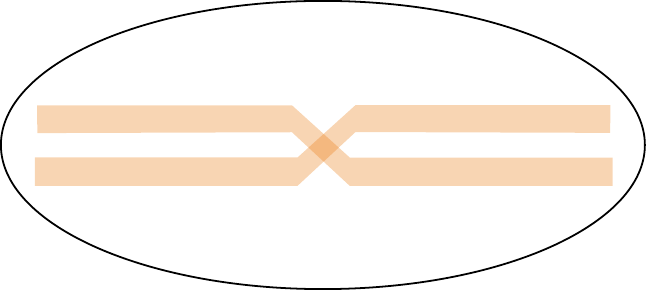}
		\caption{After routing}
	\end{subfigure}
	\caption{%
		For the proof of \cref{thm:qrtExpansionBound}, we consider a bipartite system
		consisting of $X$ and $\bar X$ with $\abs{X} \le \abs{\bar X}$.
		The subsystems $X$ and $\bar X$ consist of qubits represented by vertices and are augmented with ancilla spaces $x$ and $x'$, respectively.
		To each qubit in $X$ and $\bar X$ we associate one ancilla in $x$ and $x'$, respectively.
		We initialize each qubit-ancilla pair in a Bell state (wavy line).
		The entropy of subsystem $Xx$ is 0.
		We then perform routing to exchange $X$ with a subset of $\bar X$ (in red).
		This increases the entropy of $Xx$ to $2\abs{X}$.
		We bound the entanglement increase for each layer of gates
		by twice the maximum matching size in $\partial(\delta X)$,
		thereby lower-bounding the circuit depth and $\qrt(G)$.
	}\label{fig:routingEntanglement}
\end{figure}

We now show that a lower bound on the gate-based quantum routing number follows from \cref{lem:gateBasedStatePreparation}
by preparing an appropriate initial state.
See \cref{fig:routingEntanglement} for an illustration of the proof concept.
\begin{theorem}\label{thm:qrtExpansionBound}
	For any simple graph $G$ and partition $X \subseteq V(G)$ with $\abs{X} \le \abs{V(G)}/2$,
	\begin{equation}\label{eq:qrtExpansionBoundSimple}
		\qrnumber{G} \ge \frac{\abs{X}}{\abs{\matching{\partial X}}}
	\end{equation}
	and
	\begin{equation}\label{eq:qrtExpansionBound}
		\qrnumber{G} \ge \max\paren{\frac{2\abs{X} - \abs{\delta X}}{\abs{\matching{\partial(\delta X)}}},
		\frac{2\abs{X} - \abs{\delta \bar X}}{\abs{\matching{\partial(\delta \bar X)}}}
		}.
	\end{equation}
\end{theorem}
\begin{proof}
	We augment the subsystems $X$ and $\bar X$ with ancilla spaces $x$ and $x'$, respectively,
    with one ancilla qubit for each vertex in $X$ and $\bar X$.
	Since these ancillas are not connected with the main graph,
	they cannot help with routing.
    Each qubit and ancilla pair forms a Bell pair in the initial state $\rho_0$.
	Then the entropy $S_{Xx}(\rho_0) = S_{\bar Xx'}(\rho_0) = 0$ since the reduced state is pure.

	The gate-based quantum routing number $\qrnumber{G}$ considers the worst-case permutation of the vertices.
	So, to show a lower bound, it suffices to pick a permutation $\pi$ that routes all vertices $v \in X$ to $\bar X$ arbitrarily
	and routes $\abs{X}$ vertices $u \in \bar X$ to $X$ arbitrarily.
	Let the resulting state be our target state $\rho$.
	This gives $S_{Xx}(\rho) = S_{\bar Xx'}(\rho) = 2\abs{X}$.
	By \cref{lem:gateBasedStatePreparation}, the depth of any circuit performing this state preparation
	and routing task is lower bounded as
	\begin{equation}
		\qrnumber{G, \pi} \ge \frac{\Delta S_{Xx}}{2\abs{\matching{\partial X}}} = \frac{\abs{X}}{\abs{\matching{\partial X}}},
	\end{equation}
	proving~\cref{eq:qrtExpansionBoundSimple}.
	Similarly, \cref{cor:gateBasedVertexLB} implies
	\begin{equation}
		\qrnumber{G, \pi} \ge \frac{\Delta S_{Xx} + \Delta S_{\bar Xx'} - 2\abs{\delta X}}{2\abs{\matching{\partial(\delta X)}}} = \frac{2\abs{X} - \abs{\delta X}}{\abs{\matching{\partial(\delta X)}}}.\label{eq:qrtLB1}
	\end{equation}
	By exchanging the roles of $X$ and $\bar X$, \cref{cor:gateBasedVertexLB} also gives the lower bound
	\begin{equation}
		\qrnumber{G, \pi} \ge \frac{2\abs{X} - \abs{\delta \bar X}}{\abs{\matching{\partial(\delta \bar X)}}}.\label{eq:qrtLB2}
	\end{equation}
	Taking the maximum of \cref{eq:qrtLB1,eq:qrtLB2},
	we obtain \cref{eq:qrtExpansionBound} as required.
\end{proof}

We now show that \cref{thm:qrtExpansionBound} lower bounds the gate-based quantum routing number 
in terms of the \emph{vertex expansion} (or vertex isoperimetric number)
\begin{equation}
\label{eq:cGdefinition}
	c(G) \coloneqq \min_{X \subseteq V(G) : \abs{X} \le \abs{V(G)} /2} \frac{\abs{\delta X}}{\abs{X}},
\end{equation}
which is a well-studied property of graphs~\cite{ChungRevised}.
Intuitively, the vertex expansion lower bounds how many vertices neighbor any small enough set $X$.
Therefore, the number of vertices in the induced subgraph $G[X \cup N(X)]$,
for $N(X)$ the neighborhood of $X$,
grows (or ``expands'') by at least a factor of $1 + c(G)$.
\begin{corollary}\label{cor:qrtVertexLB}
	For any simple graph $G$,
	\begin{equation}
		\qrnumber{G} \ge \max_{X \subseteq V(G): \abs{X} \le \abs{V(G)}/2} \frac{2\abs{X}}{\abs{\delta X}} - 1 = \frac{2}{c(G)} -1.
	\end{equation}
\end{corollary}
\begin{proof}
	By maximizing over all allowed partitions $X$ in \cref{thm:qrtExpansionBound},
	choosing one branch of \cref{eq:qrtExpansionBound},
	and noting $\abs{\matching{\partial(\delta Y)}} \le \abs{\delta Y}$, for any $Y \subseteq V(G)$,
	we have
	\begin{align}
		\qrt(G) &\ge \max_{X \subseteq V(G) : \abs{X} \le \abs{V(G)}/2} \frac{2\abs{X} - \abs{\delta X}}{\abs{\matching{\partial(\delta X)}}}\\
			&\ge \max_{X \subseteq V(G) : \abs{X} \le \abs{V(G)}/2} \frac{2\abs{X} - \abs{\delta X}}{\abs{\delta X}}\\
			&= \max_{X \subseteq V(G) : \abs{X} \le \abs{V(G)}/2} \frac{2\abs{X}}{\abs{\delta X}} - 1
	\end{align}
	as required.
\end{proof}
In \cref{app:matchingExpansion}, we show that the vertex expansion and a similarly defined \emph{matching expansion},
\begin{equation}\label{eq:matchingExpansion}
	\matching{G} \coloneqq \min_{X \subseteq V(G): \abs{X} \le \abs{V(G)}/2} \frac{\abs{\matching{\partial X}}}{\abs{X}},
\end{equation}
are asymptotically equivalent, i.e., $c(G) = \bigtheta{\matching{G}}$.

A simple consequence of \cref{cor:qrtVertexLB} is that gate-based quantum routing on the star graph,
$S_n \coloneqq K_{1,n}$ (the complete bipartite graph with parts of size $1$ and $n$, as shown in \cref{fig:star}),
is no faster than classical routing up to a constant factor.
A trivial classical routing strategy has a depth upper bounded by $3n/2$,
whereas we have $c(S_n) \le 2/n$ so that $\qrnumber{S_n} \ge n-1$.
This is a consequence of the small vertex cut in the star graph.

\subsection{Hamiltonian routing}
In this section, we consider a stronger model for quantum routing,
namely using two-qubit Hamiltonians with fast local operations.
The \emph{Hamiltonian routing time},
\begin{equation}
	\hrnumber{G} \coloneqq \max_\pi \hrnumber{G, \pi}
\end{equation}
where $\pi$ is a permutation of qubits and $\hrnumber{G, \pi}$,
is the minimum evolution time, normalized so that a \swap{} gate takes time 1 (discussed below),
of some time-dependent Hamiltonian $H(t)$ that
respects the architecture constraints given by $G$
(i.e., it is 2-local and only has interactions along the edges $E(G)$)
and implements $\pi$.
Note that here we consider minimizing the time as opposed to circuit depth.

A time scale follows from a normalization condition on the two-qubit interaction strength of the Hamiltonian $H(t)$
at all times.
We can write any two-qubit local Hamiltonian in the \emph{canonical form}~\cite{Bennett2002}
\begin{equation}\label{eq:CanonicalForm}
	K \coloneqq \sum_{j \in \set{x,y,z}} \mu_j \sigma_j \otimes \sigma_j
\end{equation}
up to local unitaries,
where $\mu_x \geq \mu_y \geq \abs{\mu_z} \geq 0$, and $\sigma_x, \sigma_y, \sigma_z$ are the Pauli matrices.
We impose the condition that $\norm{K} = \sum_j \abs{\mu_j} \leq 3\pi/4$ for all interactions
in $H(t)$ at all times $t$~\cite{Bapat2020},
where $\norm{\cdot}$ is the spectral norm.
Recall that we consider a model in which local operations can be performed arbitrarily quickly.
The shortest \cnot{}  time in this model is $1/3$
and the shortest \swap{} time is $1$~\cite{Vidal2002}.
Furthermore, any two-qubit unitary takes at most time 1
since any such gate can be decomposed into at most 3 \cnot{} gates and single-qubit rotations~\cite{Vatan2004}.
Therefore, this normalization guarantees
$\hrnumber{G, \pi} \le \qrnumber{G, \pi}$ for any permutation $\pi$,
and in particular, $\hrnumber{G} \le \qrnumber{G}$.
We now show that the Hamiltonian routing time is lower bounded by the diameter of the graph over the maximum degree.

\begin{theorem}\label{thm:hqrtDiam}
	For any simple graph $G$,
	\begin{equation}
		\hrnumber{G} = \bigomega*{\frac{\diam(G)}{\max_v d_v}},
	\end{equation}
	where $d_v$ is the degree of $v \in V(G)$.
\end{theorem}
\begin{proof}
	Pick two vertices $u,v \in V(G)$ at a distance $\diam(G)$.
	In the Heisenberg evolution picture,
	routing must be able to map an $X$ operator on $u$ at time 0,
	$X_u(0)$,
	to $X_u(T)$ supported on $v$ after some time $T$.
	This means a $Z$ operator on $v$ at time 0, $Z_v$, has
	$\norm{[Z_v, X_u(T)]} = 2$.
	\textcite[Eq.~7]{Nachtergaele2006} bound this unequal time commutator
	after time $t$ by
	\begin{equation}
		\norm{[Z_v, X_u(t)]} \le 2 e^{C\abs{t} - \diam(G)},
	\end{equation}
	where
	\begin{equation}
	    C = 2^6 e \max_{w \in V(G)} \sum_{e=(w,w') \in E(G)}  \norm{H^{(e)}} \le 3 \pi 2^4 e \max_w d_w
	\end{equation}
	and $H^{(e)}$ is a two-qubit Hamiltonian term acting only on the ends of the edge $e$.
	Therefore, the time is lower bounded by $t = \bigomega*{\frac{\diam(G)}{\max_w d_w}}$.
\end{proof}

The dependence on the maximum degree is necessary when we consider a multigraph
with two vertices connected by $k$ edges.
We can then speed up any normalized interaction between the two vertices
linearly in the degree, $k$.
In particular, it is possible to implement a \swap{} in time $1/k$.
We use a similar idea to show separations between strengthened gate-based and Hamiltonian routing models in \cref{sec:separation}.
It is an open question whether a Hamiltonian routing protocol on a simple graph
can have a routing time that is upper bounded by $\littleo{\diam(G)}$.

We show that the Hamiltonian routing time can also be lower bounded by an edge cut in the graph $G$.
An edge cut partitions $G$ into two vertex subsets $X \subseteq V(G)$ and $\bar X$.
The edges leaving $X$ form the \emph{edge boundary} of $X$,
\begin{equation}
	\partial X \coloneqq \set{(x, \bar x) \in E \mid x \in X, \bar x \in \delta X} = \partial \bar X,
\end{equation}
and are an edge cut.
We define the \emph{edge expansion} (or edge isoperimetric number or Cheeger constant) as
\begin{equation}\label{eq:edgeExpansion}
	h(G) \coloneqq \min_{X \subseteq V(G) : \abs{X} \le \abs{V(G)}/2} \frac{\abs{\partial X}}{\abs{X}}.
\end{equation}
Intuitively, this corresponds to a lower bound on how many edges
leave any small enough set $X$.
Therefore, the number of edges in the induced subgraph $G[X \cup N(X)]$
grows (or ``expands'') by at least $1 + h(G)$.

In the following, we show a lower bound of $\hrnumber{G} = \bigomega{1/h(G)}$.
Because $|\partial X| \geq |\delta X|$, the edge expansion is always at least as large as the vertex expansion, i.e., $h(G) \ge c(G)$,
so this is a weaker lower bound than we showed above on gate-based quantum routing.
In particular, the star graph has $h(S_n) = \bigtheta{1}$
so our lower bound gives $\hrnumber{S_n} = \bigomega{1/h(S_n)} = \bigomega{1}$.
Since $\qrnumber{S_n} = \bigomega{n}$,
this does not rule out the possibility of a large separation between Hamiltonian and gate-based quantum routing.

% \begin{figure}[t]
% 	\centering
% 	\includeinkscape{bipartite-system}
% 	\caption{%
% 		The bipartite system consisting of subsystems $\mathcal A$ and $\mathcal B$,
% 		where $\mathcal A$ consists of a subsystem $A$ and remaining ancilla system $a$,
% 		and similarly for $\mathcal B$ and subsystem $B$ with remaining ancilla system $b$.
% 		The only interactions between $\mathcal A$ and $\mathcal B$
% 		happen by the 2-local interaction Hamiltonian $H_{AB}$ between qubits in $A$ and $B$.
% 	}\label{fig:bipartite-system}
% \end{figure}

To prove the lower bound on Hamiltonian routing,
we use the continuous analogue of STE,
the \emph{small incremental entangling} (SIE) theorem, adapted to our setting.
SIE was conjectured by Kitaev~\cite{Bravyi2007}
and first proven in~\cite{Acoleyen2013}.
\begin{lemma}[Small Incremental Entangling (SIE)]\label{lem:sie}
	Given a finite joint system $X\bar X$,
	any Hamiltonian $H$ with support only on $\delta X \delta \bar X$ and
	any initial pure state $\rho$,
	the entanglement capacity $\capacity{H}{\rho}$ is bounded as
	\begin{equation}
		\capacity{H}{\rho} \coloneqq \frac{d S_X(\rho(t))}{dt} \leq \sieConst \norm{H} \log d,
	\end{equation}
	where $\rho(t) = U(t)\rho U(t)^\dagger$ for $U(t) = e^{-iHt}$,
	$0< \sieConst \le 4$ is a constant, and $d = \min(\abs{\delta X}, \abs{\delta \bar X})$.
\end{lemma}
It is conjectured that $\sieConst =2$~\cite{Bravyi2007} but the best known bound gives $\sieConst=4$~\cite{Audenaert2014}.
No generality is lost by assuming pure states since
we can add an ancillary purification system $C$ to $X$ without loss of generality.
The resulting state on the joint system $X\bar X C$ is pure and constrained by SIE.
Since including $C$ as an ancilla can only increase the entanglement capacity
(we can always ignore it),
we see that the entanglement capacity is also bounded for mixed states on $X\bar X$.

We can derive another expression for $\capacity{H}{\rho}$ by writing
\begin{align}
	\capacity{H}{\rho} &= -\frac{d}{dt}\trace*{ \rho_X(t) \log \rho_X(t) }\\
								 &= -\trace*{\frac{d \rho_X(t)}{dt} \log \rho_X(t)}\\
								 &= i\trace*{\ptrace{\bar X}{[H,\rho]} \log \rho_X(t)},\label{eq:capacityLinear}
\end{align}
where we used the Schr\"{o}dinger equation $i\frac{d\rho}{dt} = [H,\rho]$ (setting $\hbar = 1$).
We see that the entanglement capacity is linear in $H$.

The evolution of a system with interaction graph $G$,
for any $X \subseteq V(G)$,
can be described by a Hamiltonian $H = H_X + H_{\bar X} + H_{\delta X \delta \bar X}$,
where $H_Y$ only has support on the subsystem of vertices $Y \subseteq V(G)$.
Operations local to $X$ or $\bar X$ do not generate entanglement,
so
\begin{equation}
	\capacity{H}{\rho} = \capacity{H_{\delta X \delta \bar X}}{\rho}.\label{eq:localCapacity}
\end{equation}
We can verify this by first explicitly computing
\begin{equation}
	\ptrace*{\bar X}{[H_{\bar X}, \rho]} = 0
\end{equation}
because the partial trace is cyclic on the $\bar X$ subsystem.
Second,
\begin{equation}
	\capacity{H_X}{\rho} = i\trace{[H_X, \rho_X(t)] \log \rho_X(t)} = 0
\end{equation}
because $\log \rho_X(t)$ commutes with $\rho_X(t)$ and the trace is cyclic.
By linearity, \cref{eq:localCapacity} holds, 
and we can restrict ourselves to consider only Hamiltonians of the form $H_{\delta X \delta \bar X}$.

Now we can bound the entanglement capacity of any edge cut in the graph
as specified by the edge boundary of a vertex subset $X$.
A slightly weaker result up to constant factors was proved in~\cite{Bravyi2006} by
using bounds on the entanglement capacity of bipartite product Hamiltonians~\cite{Childs2004} instead of SIE.

\begin{theorem}\label{thm:entanglementCapacityCut}
	Given any $X \subseteq V(G)$
	and any pure state $\rho$,
	the entanglement capacity of a Hamiltonian $H$ with support
	only on the joint subsystem $\delta X \delta \bar X$
	satisfies
	\begin{equation}
		\capacity{H}{\rho} = \frac{d S_X(\rho(t))}{dt} \leq \frac{3\pi\sieConst}{4} \abs{\partial X},
	\end{equation}
	for $\sieConst$ the constant of SIE.
\end{theorem}
\begin{proof}
We decompose the Hamiltonian into a sum of local terms $H = \sum_{e \in \partial X} H^{(e)}$
% The following was defined earlier but repeated for clarity.
where each $H^{(e)}$ is a two-qubit Hamiltonian acting only on the ends of the edge $e$.
By linearity,
\begin{equation}
	\capacity{H}{\rho} = \capacity*{\sum_{e \in \partial X} H^{(e)}}{\rho}
		= \sum_{e \in \partial X} \capacity*{H^{(e)}}{\rho}.
\end{equation}
We bound each term by SIE (\cref{lem:sie}):
\begin{equation}
	\sum_{e \in \partial X} \capacity*{H^{(e)}}{\rho} \le \sieConst \sum_{e\in \partial X} \norm*{H^{(e)}}.
\end{equation}
By unitary similarity (which the norm is invariant under),
we can rewrite each term in canonical form~\cref{eq:CanonicalForm}
and apply our normalization condition such that
$\sum_{e\in \partial X} \norm{H^{(e)}} \le (3\pi/4)\abs{\partial X}$.
\end{proof}

Using this relation of entanglement capacity to edge cuts in the graph,
we show a lower bound on the time to perform state preparation 
in the Hamiltonian model dependent on the edge cut.

\begin{corollary}\label{cor:hamiltonianStatePreparation}
	Given an initial pure state $\rho_0$ and target pure state $\rho$
	on a bipartite system $X \bar X$,
	define the change in entanglement entropy $\Delta S_X \coloneqq \abs{S_X(\rho) - S_X(\rho_0)}$.
	Then any Hamiltonian unitary 
	evolution from $\rho_0$ to $\rho$
	restricted by interaction graph $G$ must have evolution time
	\begin{equation}
		t \ge \frac{4}{3\pi\sieConst}\frac{\Delta S_X}{\abs{\partial X}}.
	\end{equation}
\end{corollary}
\begin{proof}
	The claim follows directly from \cref{thm:entanglementCapacityCut}.
\end{proof}

A lower bound on Hamiltonian routing follows
since routing a particular state is a special case of state preparation.
\begin{theorem}\label{thm:hqrtEdgeBound}
	For any simple graph $G$,
	\begin{equation}
	\hrnumber{G} \ge \frac{8}{3\pi\sieConst}\frac{1}{h(G)}.
	\end{equation}
\end{theorem}
\begin{proof}
	We prepare the same initial state as in \cref{thm:qrtExpansionBound},
	where we have one half of a Bell pair at each vertex $v \in V(G)$
	that is entangled with an ancilla.
	To show a lower bound, we pick some $X \subseteq V(G)$ with $\abs{X} \le \abs{V(G)}/2$
	and an associated ancilla space $x$,
	and pick a permutation $\pi$ that routes all vertices $v \in X$ to $\bar X$ arbitrarily
	and routes $\abs{X}$ vertices $u \in \bar X$ to $X$ arbitrarily.
	Let the resulting state be our target state $\rho$.
	This gives $\Delta S_{Xx} = S_{Xx}(\rho) = 2\abs{X}$.
	\cref{cor:hamiltonianStatePreparation} implies that the time to implement this state preparation and routing task is lower bounded as
	\begin{equation}\label{eq:hrnumberLBX}
		\hrnumber{G, \pi} \ge \frac{4}{3\pi\sieConst} \frac{\Delta S_{Xx}}{\abs{\partial X}} = \frac{8}{3\pi\sieConst} \frac{\abs{X}}{\abs{\partial X}}.
	\end{equation}
	We now maximize over all $X$ to lower bound the Hamiltonian routing time
	\begin{equation}
	\hrnumber{G} = \max_\pi \hrnumber{G, \pi} \ge \max_{X: \abs{X} \le \abs{V(G)}/2} \frac{8}{3\pi\sieConst} \frac{\abs{X}}{\abs{\partial X}} = \frac{8}{3\pi\sieConst} \frac{1}{h(G)}
	\end{equation}
	as claimed.
\end{proof}

\begin{figure}
	\centering
	\begin{subfigure}{0.4\textwidth}
		\centering
		\includeinkscape[height=85pt]{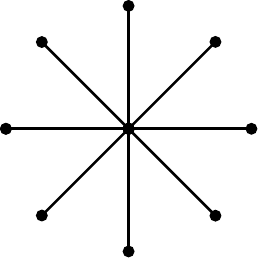}
		\caption{Star graph $S_n$, for $n=8$}\label{fig:star}
	\end{subfigure}%
	\hfill%
	\begin{subfigure}{0.6\textwidth}
		\centering
		\includeinkscape[height=85pt]{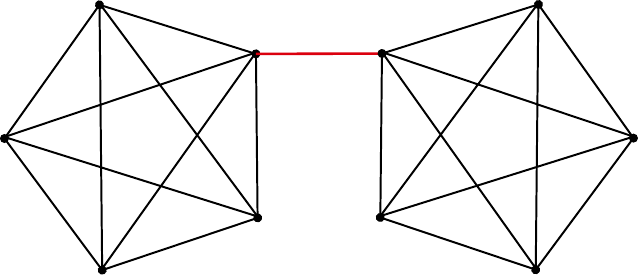}
		\caption{Barbell graph $C_{2n}$, for $n=5$}\label{fig:ebarbell}
	\end{subfigure}
	\caption{%
		Gate-based routing models cannot be separated up to a constant additive factor from classical routing on the star graph,
		but our lower bound on Hamiltonian routing is trivial in this case because of its $\bigomega{1}$ edge expansion.
		Classical and Hamiltonian routing cannot be separated on the barbell graph
		because of its $\bigo{1/n}$ edge expansion.
	}
\end{figure}

One simple example where this rules out a separation between classical and Hamiltonian routing is the \emph{barbell graph}, $C_{2n}$~\cite{Ghosh2008}.
The barbell graph consists of two complete graphs, $K_n$, connected by a single edge at some vertex in each complete graph,
as shown in \cref{fig:ebarbell}.
Since $h(C_{2n}) \le 2/n$,
\cref{thm:hqrtEdgeBound} implies the Hamiltonian routing time on this graph is lower bounded as
\begin{equation}
	\hrnumber{C_{2n}}\ge \frac{4n}{3\pi\sieConst}.
\end{equation}
By routing on its spanning tree, $\rnumber{C_{2n}} = \bigo{n}$, so classical routing is tight up to a constant factor.

An entanglement capacity bound of $\bigo{\abs{\partial X}}$, as given by \cref{thm:entanglementCapacityCut},
matches previous results on entanglement area laws for dynamics~\cite[theorem~1]{Gong2017} on lattices of constant dimension.
For graphs of superconstant degree,
the distinction between bounds on the entanglement capacity proportional to edge cuts (for Hamiltonian routing)
and vertex cuts (for gate-based quantum routing) are significant.
In general, $\abs{\partial X} \leq \abs{\delta X} \max_v d_v $.
It remains an open question whether Hamiltonian routing can be separated by a superconstant factor from gate-based quantum routing,
and in particular, if the Hamiltonian routing time can also be lower bounded by the vertex expansion $\bigomega{1/c(G)}$.
However, we show in \cref{sec:separation} that a stronger model of Hamiltonian routing can be separated from gate-based quantum routing
and its routing time cannot be lower bounded by $\bigomega{1/c(G)}$.

Another case that has been well studied is the path graph, $P_n$.
Here, the \emph{odd-even sort}~\cite{Knuth1998} gives a simple classical routing algorithm that upper bounds
the circuit depth by $n$.
A simple bound on the vertex expansion of the path graph is $c(P_n) \le 2/n$,
so $\qrnumber{P_n} \ge n - 1$, matching the diameter lower bound (\cref{thm:diamLB}) up to an additive constant.
Thus, a constant-factor improvement over classical routing on the path is only possible in the Hamiltonian routing model.
In that case,
we have $h(P_n) \le 2/n$, giving $\hrnumber{P_n} \ge 4n/(3\pi \sieConst)$.
This is slightly weaker (even if $\sieConst = 2$) than a specialized bound of $4n/(3\pi\alpha_0) \approx 0.222n$, for $\alpha_0 \approx 1.912$,
based on the entanglement capacity~\cite{Bapat2020}.
Indeed, \textcite{ReversalSort} show that $\hrnumber{P_n} \le (1-\varepsilon)n + \bigo{\log^2 n}$
for a constant $\varepsilon \approx 0.034$,
so, for large enough $n$, $\hrnumber{P_n} < \qrnumber{P_n}$ with a constant-factor speedup.

\section{Comparison with classical routing}\label{sec:classical} 
Fast classical routing algorithms are already known for some graph families~\cite{Alon1994,Childs2019}.
An example is the family of grid graphs,
which are Cartesian products of path graphs $P_{L_1} \square P_{L_2}$ with dimensions $L_1,L_2 \in \mathbb N$,
where we know $\rnumber{P_{L_1} \square P_{L_2}} \le 2L_1 + L_2$.
We can exclude a superconstant quantum advantage simply by the diameter lower bound (\cref{thm:hqrtDiam}).

In this section, we compare our quantum routing results with general bounds on classical routing.
In particular, this gives more conditions for a superpolynomial separation.
Our results and proofs are generalizations of results in~\cite{Alon1994}
from regular graphs to irregular graphs.

In classical routing, we route a permutation $\pi$ in multiple time steps.
We first assign to each vertex $v$ a \emph{token} labeled $\pi(v)$.
Then, in each time step, we perform \swap{} gates on neighboring vertices
to exchange their tokens
with the constraint that each vertex participates in at most one \swap{}.
Routing terminates when all tokens have been moved to their destination vertices.
The difficulty of classical routing on $G$ is characterized by the \emph{routing number}~\cite{Alon1994}
\begin{equation}
	\rnumber{G} \coloneqq \max_{\pi} \rnumber{G, \pi},
	\label{eq:rnumber}
\end{equation}
where $\rnumber{G, \pi}$ is defined as the minimal number of time steps needed to implement the permutation $\pi$.
Since gate-based routing generalizes \swap{}-based routing, $\qrnumber{G, \pi} \le \rnumber{G, \pi}$ for any permutation $\pi$, and in particular, $\qrnumber{G} \le \rnumber{G}$.

% As a point of comparison, we review several bounds on the routing number.
% For any tree graph the routing number is bounded by $3n/2 + \bigo{\log n}$~\cite{Zhang1999},
% giving us the same bound for any connected graph by its minimum-spanning tree.
% The diameter also bounds the routing number as $\rt(G) \ge \text{diam}(G)$
% since it takes that long to \swap{} from one side to the other of a graph~\cite{Alon1994}.
% Finally, a vertex cut $C$ that partition $G$ into disconnected components $A$ and $B$ makes a bottleneck~\cite{Alon1994}
% \begin{equation}
% 	\rt(G) \ge \frac{2}{|C|}\min\set*{\abs{A}, \abs{B}}.\label{eq:rtvcut}
% \end{equation}
% Assuming, w.l.o.g., $\abs{A} \leq \abs{B}$,
% \cref{eq:rtvcut} follows from considering a permutation that maps all vertices in $A$ to $B$ injectively.
% Then, at most $\abs{C}$ swaps may be performed to move tokens into and then out of $C$ per step,
% thus taking 2 steps to swap from $A$ to $B$ or vice versa.
% (The function of $C$ is similar to $\delta X$ that we introduce later and in \cref{fig:biparted-graph}.)

\subsection{General classical routing}
We now describe a classical routing algorithm that performs \swap{}s along a set of walks (connecting each token with its destination) that are close to random.
The number of \swap{}s that act on the same vertices at the same time is bounded from above by the inverse spectral gap of the (normalized) graph Laplacian,
leading to high parallelism in graphs with large spectral gap.

The set of vertices is isomorphic to an integer labeling, $V(G) \cong [n]$, so we identify each $v \in V(G)$ with a unique integer index.
Let the adjacency matrix $A$ have entries
\begin{equation}
	A_{uv} = \begin{cases*}
		1 & if $(v,u) \in E(G)$,\\
		0 & otherwise,
	\end{cases*}
\end{equation}
for $v,u \in V(G)$,
and let the diagonal matrix $T$ have entries $T_{vv} = d_v$ and 0 otherwise,
for $d_v = (A\vec 1)_v$ the degree of $v$ and $\vec 1$ the all-ones vector.
Then the (normalized) graph Laplacian is $\mathcal L \coloneqq \idm - T^{-1/2}AT^{-1/2}$.
The Laplacian is symmetric and positive semidefinite~\cite{ChungRevised}
and has a 0 eigenvalue for the eigenvector $T^{1/2}\vec 1$.
Let the \emph{spectral gap}, $\lambda(G)$, be the smallest non-zero eigenvalue of $\mathcal L$.

In this section,
we assume $n \ge 2$ and show a general bound on the routing number
without attempting to minimize the constants.
Let $v_1v_2\dots$ denote a walk on the vertices $v_i\in V(G)$ that passes through $v_i$ at time step $i$.
We consider memoryless random walks with transition probabilities
denoted by $P_{vu} = \probability{x_{i+1} = v \mid x_i=u}$.
These probabilities form the transition matrix $P$ of the random walk on $G$.
We choose the \emph{lazy random walk} $P = (\idm + AT^{-1})/2$, i.e.,
\begin{equation}
	P_{vu} = \begin{cases*}
		1/2 & if u = v,\\
		1/(2d_u) & if $(u,v) \in E(G)$,\\
		0 & otherwise.
	\end{cases*}
\end{equation}
In the following, we will refer to lazy random walks simply as random walks.
Note that we default to right multiplication with the transition matrix
so our probability distributions can be interpreted as column vectors.
Therefore, the probability that a random walk starting at $u$ is at $v$ after $i \in \mathbb N$ steps is given by
\begin{equation}
	\probability{x_i = v \mid x_0 = u} =\vec e(v)\tran P^i \vec e(u),\label{eq:pathIntersect}
\end{equation}
where $\vec e(v)$ is the column vector with a 1 in position $v$ and 0 otherwise.
The stationary distribution of the walk $P$ is $\vec \pi \coloneqq T\vec 1$ since $PT\vec 1 = T\vec 1$.

\begin{figure}
	\centering
	\includeinkscape{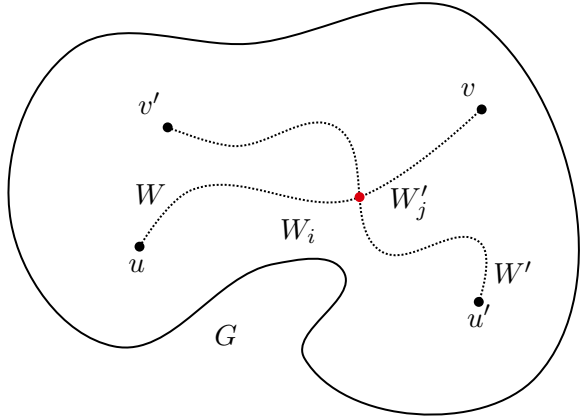}
	\caption{%
		Shown are walks $W$ from $u$ to $v$
		and $W'$ from $u'$ to $v'$.
		Walks may intersect at a vertex (red)
		such that the $i$th location of $W$ and
		the $j$th location of $W'$ are the same,
		i.e., $W_i = W'_j$.
		We say that $W$ and $W'$ \emph{interfere} if there exist $i, j$ with $\abs{i-j} \le 1$ such that $W_i = W'_j$.
		We show that, with high probability,
		there exists a set of walks for a permutation $\sigma$ of order two on $V(G)$
		that connect $v$ to $\sigma(v)$ such that the number of interfering walks can be bounded.
		\textsc{Swap}s along sets of walks that do not interfere with each other
		significantly parallelize the routing process.
	}\label{fig:randomWalks}
\end{figure}

We first define a useful notion of interference between walks. 
\begin{definition}[Interfering walks, \cref{fig:randomWalks}]
Two walks $W$ and $W'$ are said to \emph{interfere} if $W_i = W'_j$ for some $i, j\in\mathbb{N}$ with $|i-j| \le 1$.
\end{definition}
The condition that $\abs{i-j} > 1$ ensures that tokens can be swapped along $W$ in parallel with token swaps along $W'$,
namely a token being swapped along $W$ has \swap{}s that overlap at a location for two time steps.

Now, let us perform a simple random walk of a given length starting at each vertex $u$ and call this walk $W(u)$.
We show that, with high probability and for sufficiently long walks, the number of walks that interfere with a given walk can be bounded from above.
This is a generalization of \cite[lemma~2]{Alon1994} to irregular graphs,
where we explicitly analyze the dependence on the degree.
In particular, the entries of $T/\min_v d_v$ are bounded from above by the \emph{degree ratio}
\begin{equation}
	d_* \coloneqq \frac{\max_v d_v}{\min_v d_v}.
\end{equation}

\begin{lemma}\label{lem:NonintersectingWalk}
	Let $G$ be a connected simple graph on $n$ vertices and suppose $l\ge \ln(n) / \lambda(G)$.
	For every $v\in V(G)$,
	let $W(v)$ denote a random walk of length $l$ starting at vertex $v$.
	Let $I(v)$ denote the total number of other walks $W(u)$
	that interfere with $W(v)$.
	Then with probability at most $n^{-20}$ there is a vertex $v\in V(G)$ with $I(v) > 30ld_*$.
\end{lemma}
\begin{proof}
	We wish to bound $I(v)$ for any $v\in V(G)$.
	We introduce an indicator random variable
	depending on the random walks $W(v)$:
	\begin{equation}
		X_{uv} \coloneqq \begin{cases*}
			1 & if $W(u)$ and $W(v)$ interfere,\\
			0 & otherwise.
		\end{cases*}
	\end{equation}
	We include the random walk starting at $v$ in the total which only increases the expectation of $I(v)$.
	By summing over $u \in V(G)$ and including $v$,
	the expected value of $I(v)$ over random walks is bounded by
	\begin{align}
		\expectation*{I(v)}
			&\le \expectation*{\sum_u X_{uv}} = \sum_{u \ne v} \probability*{X_{uv} = 1} \\
			&= \sum_u \probability*{\bigvee_{i\in [l]} \bigvee_{j:\abs{i-j} \le 1} {W(v)}_i = {W(u)}_j} \\
			&\le \sum_u \sum_{i\in [l]} \sum_{j: \abs{i-j} \le 1} \probability*{{W(v)}_i = {W(u)}_j}.\label{eq:pathIntersect3}
	\end{align}
	Using \cref{eq:pathIntersect}, we have
	\begin{align}
		\sum_u \sum_{\substack{i \in [l] \\ j:\abs{i-j} \le 1}}\probability*{{W(v)}_i = {W(u)}_j}
			&= \sum_u \sum_{i\in [l]} {\vec e(W(v)_i)}\tran \sum_{j : \abs{i-j} \le 1} P^j \vec e(u) \\
			&= \sum_{i\in [l]} {\vec{e}(W(v)_i)}\tran \sum_{j: \abs{i-j} \le 1} P^j \vec 1.
	\end{align}
	The transposed vector $e(W(v)_i)\tran P^j$ on the right-hand side has non-negative entries for all $i, j$,
	therefore an upper bound follows from substituting $\vec 1$ by the entrywise larger vector $\vec \pi/\min_v d_v$ as
	\begin{equation}
			\sum_{i\in [l]} {\vec{e}(W(v)_i)}\tran \sum_{j: \abs{i-j} \le 1} P^j \vec 1
			\le  \sum_{i\in [l]} \vec e(W(v)_i)\tran \sum_{j:\abs{i-j} \le 1} P^j \frac{\vec \pi}{\min_v d_v}.\label{eq:irregularDegrees}\\
	\end{equation}
	The distribution $\vec \pi$ is stationary under the walk $P$, so
	\begin{align}
			\sum_{i\in [l]} \vec e(W(v)_i)\tran \sum_{j:\abs{i-j} \le 1} P^j \frac{\vec \pi}{\min_v d_v}
			% the \le comes from the edge cases where i=0 or i=l.
			&\le 3\sum_{i\in [l]}\vec e(W(v)_i)\tran \frac{\vec \pi}{\min_v d_v} \\
			&\le  3ld_*,
	\end{align}
	since $T_{uu} \le \max_v d_v$ for any $u \in V(G)$.
	Therefore $\expectation{I(v)} \le 3ld^*$.

	We now bound the tail probability of $I(v)$.
	We use the multiplicative Chernoff bound,
	which states that for a random variable $Y = \sum_i Y_i$ with mean $\mu$
	where the $Y_i$ are independent random variables,
	$\probability*{X>(1+\delta)\mu}\le \left(e^\delta/(1+\delta)^{1+\delta}\right)^{\mu}$ for any $\delta>0$.
	We see that the Chernoff bound applies to $I(v) \le \sum_u X_{uv}$ since the walks $W(u)$ are independent
	(note that they may depend on $v$).
	Applying the Chernoff bound with $\delta = 9$, we have
	\begin{equation}
		\probability*{I(v) > 30ld_*} \le \probability*{I(v) > 10\expectation*{I(v)}}
          < e^{3\ln(e^9/10^{10})ld_*}.
    \end{equation}
    Given that $3\ln(e^9/10^{10}) < -42$
    and using the lemma's assumption that $l \ge \ln(n)/\lambda(G)$,
    we obtain
    \begin{equation}
        e^{3\ln(e^9/10^{10})ld_*} < e^{-42ld_*} \le n^{-42d_*/\lambda(G)}.
    \end{equation}
	We lower bound $d_* \ge 1$
	and $1/\lambda(G) \ge 1-1/n \ge 1/2$~\cite[lemma~1.7]{ChungRevised} (we assumed $n \ge 2$)
	to obtain
  	\begin{align}
        n^{-42d_*/\lambda(G)} &\le n^{-21}.
	\end{align}
	Since there are $n$ vertices, the probability that there exists a vertex $v$ with $I(v) > 30ld_*$
	is at most $n^{-20}$.
	The lemma follows from the contrapositive.
\end{proof}

We can ``glue'' together pairs of random walks starting at
the $k/2$ pairs of vertices that are mapped to each other in
a permutation of order two
to obtain a set of $k$ glued walks.
We show that, with high probability, no glued walk in this set will have many interfering other glued walks.
This is an adaptation of \cite[lemma~3]{Alon1994} to irregular graph.

\begin{lemma}\label{lem:NonintersectingWalks}
	Let $G$ be a simple connected graph on $n$ vertices, let $\sigma$ be a permutation of order two on $V(G)$
	with $k$ vertices $v \in V(G)$ such that $\sigma(v) \ne v$,
	and let $l=\frac{20}{\lambda(G)}\ln n$.
	Then there is a set of $k$ walks $W(v)$ of length $2l$,
	where both $W(v)$ and $W(\sigma(v))$ have endpoints
	$v$ and $\sigma(v)$ and traverse the same edges (in opposite directions),
	satisfying the following:
	if $I(v)$ denotes the total number of other walks $W(u)$ that interfere with either $W(v)$ or $W(\sigma(v))$,
    then $I(v) < 120ld_*$ for all $v$ with probability at least $1 - \bigo{n^{-7}}$.
\end{lemma}

\begin{proof}
    We first show the existence of $k$ conditioned random walks (defined below) of length $l$,
    one for each vertex $v \in V(G)$ with $\sigma(v) \ne v$,
    that are close to random walks.%~\cite{Broder1994}.
	
	Define the probability of an \emph{open random walk} starting at $v$ and ending at a random (not a priori specified) vertex $w \in V(G)$,
	after $t \in \mathbb N$ steps,
	as
	\begin{align}
		P_{v}^{(t)}(w) &\coloneqq \probability*{{W(v)}_t = w} \\
					   &= \sum_{v_2,\dots,v_{t-1} \in V(G)} \probability*{W(v) = (v,v_2,\dots,v_{t-1},w)}.
	\end{align}
	%Note in the above that $v$ is fixed while $w$ is a random variable.
	We now define the \emph{relative pointwise distance}, $\Delta\colon V(G) \to \mathbb R$,
	of $P_{v}^{(t)}(w)$ to the stationary distribution $\vec \pi$ as~\cite{ChungRevised}
	\begin{equation}
		\Delta(t) \coloneqq \max_{v,w} \frac{\abs*{P_{v}^{(t)}(w) - \vec \pi(w)}}{\vec \pi(w)}.
	\end{equation}
	All random walks of length $l$ are close to stationary with respect to the relative pointwise distance
	since~\cite[theorem~1.16]{ChungRevised}\footnote{%
		\textcite[theorem~1.16]{ChungRevised} shows a slightly stronger version for a modified lazy random walk,
		which has different transition probability to remain at the same vertex.
	}
	\begin{equation}
			\Delta(l)
			\le  2e^{-\frac{l\lambda(G)}{2}} \frac{\abs{E(G)}}{\min_x d_x}
			=  2n^{-10} \frac{\abs{E(G)}}{\min_x d_x}
				<  2n^{-8},
	\end{equation}
	where we used $|E(G)| < n^2$, and $\min_x d_x \ge 1$.
	
	Now we compare the statistics of an open random walk $W(v)$ of length $l$
	to a \emph{conditioned random walk} where we condition on the last vertex being $w \in V(G)$,
	which is sampled according to the stationary distribution.
	The probability of a particular open random walk $W(v)$ can be related to
	that of the conditioned random walk by
	\begin{align}
		\label{eq:OpenVsConditionedWalk1}	\probability*{W(v)}
			&= \sum_w \probability*{W(v) \,\middle|\, {W(v)}_l = w} P_{v}^{(l)}(w)\\
			\label{eq:OpenVsConditionedWalk2}
			&\le  \sum_w \probability*{W(v) \,\middle|\, {W(v)}_l = w} (1+2n^{-8})\vec \pi(w)\\
			\label{eq:OpenVsConditionedWalk3}
		  	&= 2n^{-8} + \sum_w \probability*{W(v) \,\middle|\, W(v)_l = w} \vec \pi(w),
	\end{align}
  	and a corresponding lower bound can be derived similarly.
	Therefore, for large $n$ a conditioned random walk has vanishing deviation from an open random walk.

    For all vertices $v$ with $\sigma(v) \ne v$, we now condition $W(v)$ and $W(\sigma(v))$ to have the same terminal vertex $w$,
    which is sampled once from $\vec\pi$.
    We call the combined walk of $W(v)$ followed by the reverse of $W(\sigma(v))$ the \emph{glued walk} for $v$.
    There are $k$ glued walks that connect $v$ to $\sigma(v)$ in pairs.
    
    Finally, we bound the number of interfering glued walks with high probability by applying \cref{lem:NonintersectingWalk}.
    We first arbitrarily partition the graph into $X \subseteq V(G)$ and $\bar X$
    such that the vertices in each pair $(v,\sigma(v))$ lie in different parts.
    For any conditioned random walk $W(v)$, we can write the number of walk interferences $I(v)$ as a sum of two random variables,
    $I(v)_X$ and $I(v)_{\bar X}$,
    defined as the number of interferences with conditioned random walks $W(u)$ with $u$ from $X$ and $\bar X$, respectively,
    excluding $W(\sigma(v))$.
    Note that since
    the terminal vertices of each conditioned random walk in $X$ (or $\bar X$) are sampled independently,
    we can apply~\cref{lem:NonintersectingWalk} individually to $I(v)_X$ and $I(v)_{\bar X}$
    (taking advantage of~\eqref{eq:OpenVsConditionedWalk3}).
    We can then similarly bound $I(\sigma(v))_X$ and $I(\sigma(v))_{\bar X}$.
    The number of glued walks that interfere with a given $(v,\sigma(v))$ glued walk
    is a random variable bounded above by the sum 
    \begin{equation}
        I(v) + I(\sigma(v)) = I(v)_X + I(v)_{\bar X} + I(\sigma(v))_X + I(\sigma(v))_{\bar X}
    \end{equation}
    that can now be bounded from above by \cref{lem:NonintersectingWalk}.
    This gives $I(v) + I(\sigma(v)) < 120ld_*$ with probability at least $1-\bigo{n^{-7}}$.
\end{proof}

The existence of walks between opposite ends of an order-two permutation with few intersections
leads to a classical routing algorithm that
divides the walks into disjoint sets that do not intersect.
This adapts \cite[theorem~4.10]{ChungRevised} to the irregular graph setting using our previous lemmas.

\begin{theorem}\label{thm:OrderTwoRouting}
	Let $\sigma$ denote a permutation of order two on the vertex set of a connected graph $G$.
	Then, for $l=20\ln(n)/\lambda(G)$,
	\begin{equation}
		\rnumber*{G,\sigma} \le 2l\left(120ld_* + 1\right) = \bigo*{\frac{d_*}{\lambda(G)^2} \log^2 n}.
	\end{equation}
\end{theorem}
\begin{proof}
	Let $W(v)$ be a system of walks of length $2l$ satisfying \cref{lem:NonintersectingWalks}.
	Let $H$ be the graph whose vertices are the walks $W(v)$
	and in which $W(v)$ and $W(u)$ are adjacent if there exist two indices $0\le i, j \le l, \abs{i-j} \le 1$,
	so $W(v)_i = W(u)_j$ or $W(v)_i = W(u)_{2l-j}$.
	By \cref{lem:NonintersectingWalks}, the maximum degree of $H$ is at most $120ld_*$ with high probability,
	hence by Brooks' theorem it is vertex colorable with at most $120ld_* + 1$ colors.
	We can therefore divide the walks $W(v)$ into at most $120ld_*+1$ sets of independent walks
	of length $2l$.

	We now present the routing algorithm.
	For each set of independent walks we sequentially do the following.
	For step $i$ with $1\le i\le l$,
	we flip tokens along the edges numbered $i$ and $2l-1-i$ in each of the walks.
	After $l$ steps, the tokens at either end of the walk have been exchanged
	and the tokens not involved in any walk have not moved.
	After repeating this for all independent sets,
	all tokens have reached their destinations.

	Since this routing algorithm succeeds with positive probability,
	there exists an algorithm achieving the claimed routing number.
\end{proof}
Now we generalize to all permutations.
\begin{corollary}\label{cor:GeneralRouting}
	For every connected simple graph $G$
	and $l=20\ln(n)/\lambda(G)$,
	\begin{equation}
		\rnumber*{G} \le 4l\left(120ld_* + 1\right) = \bigo*{\frac{d_*}{\lambda(G)^2} \log^2 n}.
	\end{equation}
\end{corollary}
\begin{proof}
	Any permutation of $V(G)$ can be written as a product of two permutations of order two.
	Use \cref{thm:OrderTwoRouting} to route each sequentially to obtain the result.
\end{proof}

To the best of our knowledge,
\cref{cor:GeneralRouting} provides novel upper bounds for certain irregular graphs.
Of particular interest are irregular graphs where $d_*/\lambda(G)^2 = \littleo{n}$.
One such example is an Erd\"{o}s-R\'{e}nyi graph $G_{n,p}$,
which is an $n$-vertex graph where each edge is independently present with some probability $p$.
\citeauthor{Hoffman2019}~\cite{Hoffman2019} showed that for $p \ge (1+\delta)\log n/n$, for constant $\delta > 0$,
there is a constant $C(\delta)$ such that
\begin{equation}
	\abs{1-\lambda(G)} < \frac{C(\delta)}{\sqrt{p(n-1)}} = \bigo*{\frac{1}{\sqrt n}}
\end{equation}
with high probability.
Thus, we have that $\lambda(G) = \bigomega{1}$ with high probability
for such $p$ and large enough $n$.
Moreover, the degree ratio $d_* \to 1$ for $n \to \infty$ with high probability,
though it does not exactly equal 1 for finite $n$,
giving some irregularity.
Under these conditions, \cref{cor:GeneralRouting} shows that $\rnumber{G_{n,p}} = \bigo{\log^2 n}$ with high probability.

\subsection{Conditions for a superpolynomial separation}
To compare our upper bound on the routing number
and the Hamiltonian routing time lower bound,
we bound the Hamiltonian routing time in terms of the spectral gap.
We use the \emph{Cheeger inequality}~\cite{Cheeger1971,ChungRevised} that we state here without proof.
\begin{lemma}[Cheeger inequality]\label{lem:CheegerInequality}
	For any connected graph $G$,
	\begin{equation}
		2h_G \ge \lambda(G) > \frac{h_G^2}{2},
	\end{equation}
	where the Cheeger constant is
	\begin{equation}
		h_G \coloneqq \min_{X \subseteq V(G) : \abs{X} \le \abs{V(G)}/2} \frac{\abs{\partial X}}{\sum_{x \in X} d_x}.
	\end{equation}
\end{lemma}

The edge expansion $h(G)$ relates to $h_G$ as
\begin{equation}
	h(G) = \min_X \frac{\abs{\partial X}}{\abs{X}} \le \min_X \frac{\abs{\partial X}}{\sum_{x \in X} d_x}\max_v d_v  = h_G \max_v d_v,\label{eq:CheegerConstantConversion}
\end{equation}
where $X \subseteq V(G)$ and $\abs{X} \le \abs{V(G)}/2$.
We now rewrite the Hamiltonian routing time lower bound, \cref{thm:hqrtEdgeBound}, in terms of the spectral gap.
\begin{lemma}\label{lem:QRoutingSpectral}
	For a connected graph $G$,
	\begin{equation}
		\hrnumber{G} \ge \frac{8}{3\pi \sieConst \max_v d_v} \sqrt{\frac{1}{2\lambda(G)}}.
	\end{equation}
\end{lemma}
\begin{proof}
	Using~\cref{thm:hqrtEdgeBound,eq:CheegerConstantConversion,lem:CheegerInequality}, we have
	\begin{align}
		\hrnumber{G} &\ge \frac{8}{3\pi \sieConst \cdot h(G)} \\
					 &\ge \frac{8}{3\pi \sieConst \cdot h_G\max_v d_v}\\
					 &> \frac{8}{3\pi \sieConst \max_v d_v} \sqrt{\frac{1}{2\lambda(G)}}
	\end{align}
	as claimed.
\end{proof}

A simple way to bound the slowdown when a classical routing algorithm
is used instead of a Hamiltonian routing algorithm is the ratio of the routing times.
By \cref{cor:GeneralRouting,lem:QRoutingSpectral}, we have
\begin{equation}
	\frac{\rnumber{G}}{\hrnumber{G}} = \bigo*{\frac{d_* \max_v d_v}{\lambda(G)^{3/2}} \log^2 n}.\label{eq:advantage}
\end{equation}
By routing on a spanning tree of $G$, we have $\rnumber{G} = \bigo{n}$~\cite{Alon1994},
and, trivially, $\hrnumber{G} = \bigomega{1}$.
Therefore, \cref{eq:advantage} is nontrivial if
\begin{equation}
	d_* \max_v d_v = \littleo*{\frac{n}{\log^2 n}\lambda(G)^{3/2}}.
\end{equation}

Moreover, it is possible to bound the routing number
by a polynomial in the Hamiltonian routing time when $\lambda(G)$ is sufficiently small.
\begin{corollary}\label{cor:noSeparation}
	For a simple connected graph $G$,
	$\rnumber{G} = \bigo{\poly{\hrnumber{G}}}$,
	where $\poly{x}$ is a polynomial in $x$,
	if
	\begin{equation}\label{eq:noSeparationCondition}
		\frac{1}{h(G)} =
		\bigomega*{\poly*{d_*, \frac{1}{\lambda(G)}, \log n}}.
	\end{equation}
\end{corollary}
\begin{proof}
	We wish to understand when $\rnumber{G}$ is some polynomial of $\qrnumber{G}$,
	i.e., $\rnumber{G} = \bigo{\hrnumber{G}^k}$ for constant $k \ge 1$.
	By \cref{cor:GeneralRouting,thm:hqrtEdgeBound},
	this is the case if
	\begin{equation}
		\frac{\log \rnumber{G}}{\log \hrnumber{G}} = \bigo*{\frac{\log\left(\frac{d_*}{\lambda(G)}\log n\right)}{\log \frac{1}{h(G)}}}
	\end{equation}
	is bounded by a constant.
	\Cref{eq:noSeparationCondition} is
	a sufficient condition for this to hold.
\end{proof}

Similarly, we can use Cheeger's inequality and the diameter lower bound
to obtain conditions for polynomially relating the routing number and the Hamiltonian routing time.
\begin{theorem}\label{thm:noSeparation}
	For a simple connected graph $G$,
	$\rnumber{G} = \bigo{\poly{\hrnumber{G}}}$ if
	\begin{equation}
		\max\left(\frac{1}{\lambda(G) \max_v d_v}, \diam(G)\right) =
		\bigomega*{\poly*{d_*, \frac{1}{\lambda(G)}, \log n}}.\label{eq:noSeparationSufficient}
	\end{equation}
\end{theorem}
\begin{proof}
By \cref{cor:GeneralRouting}, \cref{lem:QRoutingSpectral}, and the diameter lower bound,
this happens when
\begin{equation}
	\frac{\log \rnumber{G}}{\log \hrnumber{G}} = \bigo*{\frac{\log\left(\frac{d_*}{\lambda(G)}\log n\right)}{\log \max\left(\frac{1}{\lambda(G) \max_v d_v}, \diam(G)\right)}}
\end{equation}
can be upper bounded by a constant.
\Cref{eq:noSeparationSufficient} is a sufficient condition.
\end{proof}

We define a \emph{separation} between the routing number and the Hamiltonian routing time
as a function $f\colon \mathbb R \to \mathbb R$ such that
\begin{equation}
    \rnumber{G} = \bigomega{f(\hrnumber{G})}.
\end{equation}
For example, a quadratic separation corresponds to $f(x) = x^2$.
\Cref{thm:noSeparation} bounds the separation to polynomial
for trivial cases such as graphs with $\diam(G) = \bigomega{n^c}$ for $c > 0$
since $\rnumber{G} = \bigo{n}$ from routing on a spanning tree~\cite{Zhang1999}.
Furthermore, there is no superpolynomial separation for bounded-degree graphs $G$,
since $\diam(G) = \bigomega{\log n}$ such that \cref{eq:noSeparationSufficient} simplifies to
\begin{equation}
	\max\left(\frac{1}{\lambda(G)}, \log n\right) = \bigomega*{\poly*{\frac{1}{\lambda(G)}, \log n}},
\end{equation}
which is always satisfied.
In particular, the separation is quadratic in the case $\lambda(G) = \bigomega{1}$.

There are families of graphs where \cref{thm:noSeparation} limits the separation to polynomial
that cannot be obtained from the diameter lower bound on Hamiltonian routing, \cref{thm:hqrtDiam},
and results for classical routing on regular graphs~\cite{Alon1994}.
An example is given by a family of irregular bounded-degree graphs constructed by \textcite{Raghavan1994} with arbitrary $h(G)$.
The diameter of this graph family is $\bigtheta{1/h(G)}$.
Thus, when we pick a subpolynomial $1/h(G)$, i.e., $1/h(G) = \littleo{n^c}$ for all constant $c > 0$,
\cref{thm:noSeparation} implies a polynomial limit on the separation
that does not follow from the diameter lower bound on Hamiltonian routing.

However, there are graphs with large spectral gap but unbounded degree
that are not restricted to a polynomial separation by \cref{thm:noSeparation}.
The star graph $S_n$ has $\lambda(S_n) = 1$~\cite{ChungRevised}
but is a poor vertex expander since $c(S_n) = \bigo{n^{-1}}$,
giving $\qrnumber{S_n} = \bigtheta{n}$.
We cannot exclude the possibility that Hamiltonian quantum routing could exhibit a superpolynomial separation in this case,
since our lower bound on $\hrnumber{S_n}$ from \cref{thm:hqrtEdgeBound} is trivial.
We take a first step toward exhibiting separations in the next section.

\section{Toward a separation}\label{sec:separation}
We have given necessary conditions for a superpolynomial separation between Hamiltonian and classical routing,
but we are not even aware of any superconstant separation.
In this section, we describe separations in stronger routing models.

\subsection{A quadratic separation with ancillas}
First we show that such a separation is possible in a variant of the Hamiltonian routing model that allows local ancilla qubits. The main idea of our construction is to consider a vertex bottleneck.
This argument also shows that the Hamiltonian routing with ancillas
cannot be lower bounded by $\bigomega{1/c(G)}$.

\begin{figure}
	\centering
	\includeinkscape{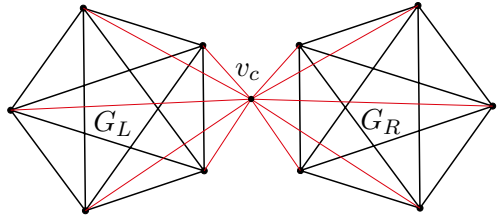}
	\caption{%
		The vertex barbell graph $B_{2n}$, for $n=5$,
		consisting of two complete graphs connected through an additional vertex.
		We have $\lambda(G) \le 2/n$ by \cref{lem:CheegerInequality}.
	}\label{fig:barbell}
\end{figure}

We show a separation on a graph $B_{2n}$, for $n \in \mathbb N$,
that we call the \emph{vertex barbell graph} (see \cref{fig:barbell}).
It consists of two complete graphs, $G_L$ and $G_R$, of $n$ vertices each and a central vertex $v_c$
where each complete graph is fully connected with $v_c$, forming two complete graphs of size $n+1$ joined at a vertex.
We have $\qrnumber{B_{2n}} = \Theta(n)$:
\cref{cor:qrtVertexLB} with $1/c(B_{2n}) \ge n$ implies the lower bound
and a trivial \swap{} routing strategy implies the upper bound.
The Hamiltonian routing time is not similarly bounded
because \cref{thm:hqrtEdgeBound} only implies a trivial lower bound, $\hrnumber{B_{2n}} = \bigomega{1}$,
since $1/h(B_{2n}) = \bigomega{1}$,
making the vertex barbell graph a potential candidate for a separation.

We are able to show a separation in the stronger model of \emph{Hamiltonian routing with ancillas}.
This model is based on Hamiltonian routing with two additional assumptions:
\begin{ienumerate}%
\item each qubit has one associated ancilla qubit available,
and \item the ancilla can perform a \swap{} with its associated qubit in negligible time.
% \footnote{%
% 	Instead of fast \swap{}s, we could consider the graph $B_{4n}$, where half the qubits in $G_L$ and $G_R$ are ancilla.
% 	In Hamiltonian routing with ancilla, the ancilla are only connected locally.%
% }
\end{ienumerate}
We denote the \emph{Hamiltonian routing time with ancilla} as
\begin{equation}
	\hrAnumber{G} \coloneqq \max_\pi \hrAnumber{G, \pi},
\end{equation}
where $\hrAnumber{G, \pi}$ is the routing time
in the Hamiltonian routing with ancilla model of $\pi$ on graph $G$.
As a point of comparison,
we may define a modified gate-based quantum routing number $\qrt_{\mathrm{a}}(G)$ analogously.
Due to the vertex bottleneck, we still have $\qrt_{\mathrm{a}}(B_{2n}) = \Theta(n)$.

We can use a protocol for fast state transfer~\cite{Guo2019}
to implement Hamiltonian routing with ancillas for the hard case of routing on $B_{2n}$.
\begin{theorem}\label{thm:barbellSwap}
	Given a vertex barbell graph $B_{2n}$ and a permutation $\sigma$
	that permutes all vertices from $G_L$ to $G_R$ and vice versa,
	we have
	\begin{equation}
		\hrAnumber{B_{2n}, \sigma} = \bigo{\sqrt{n}}.
	\end{equation}
\end{theorem}

\begin{figure}
	\centering
	\includeinkscape{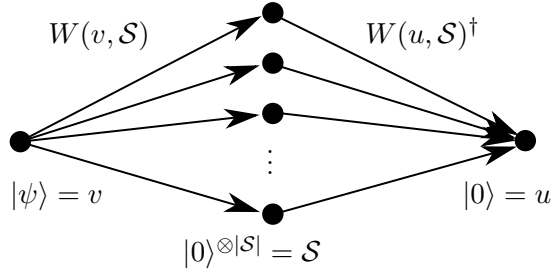}
	\caption{%
		In our routing protocol for the vertex barbell graph, we transfer the state $\ket \psi$ on qubit $v$ to $u$
		by using the intermediate qubits $\mathcal S$ as ancillas (in the $\ket{0}^{\otimes \abs{\mathcal S}}$ state).
		The operation $W(v, \mathcal S)$ encodes $\ket\psi$ in a subspace spanned by $\ket{0}^{\otimes \abs{\mathcal S}}$ and the W-state
		in time $\pi/(2\sqrt{\abs{\mathcal S}})$~\cite{Guo2019}.
		Since this procedure is unitary, we can use its inverse to transfer the state to $u$.
		We repeat this procedure in $G_R$ to transfer the state to its destination.
	}\label{fig:state-transfer}
\end{figure}

\begin{proof}
	We define a Hamiltonian to construct a W-state~\cite{Guo2019},
	\begin{equation}
		W(x,\mathcal S) \coloneqq \sum_{v \in \mathcal S} c^\dagger_x c_v + \text{h.c.},
	\end{equation}
	where $\mathcal S \subseteq V(B_{2n})$, $x \in V(B_{2n}) \setminus \mathcal S$,
	and $c_y = \ket{0}_y\bra{1}_y$ (resp.\ $c_y^\dagger$) are annihilation (resp.\ creation) operators acting on qubit $y \in V(B_{2n})$.
	Evolving for time $\pi / (2\sqrt{\abs{\mathcal S}})$
	with initial state $\ket\psi = a_0 \ket{0} + a_1 \ket 1$ on $x$,
	we have
	\begin{equation}
		e^{-i W(x,\mathcal S) T} (a_0\ket{0}_x + a_1\ket{1}_x) \ket{0}_{\mathcal S} = \ket{0}_x (a_0\ket{0}_{\mathcal S}+ a_1\ket{W}_{\mathcal S}),
		\label{eq:wstateencoding}
	\end{equation}
	where $\ket{W} \coloneqq \frac{1}{\sqrt{\abs{\mathcal S}}} \sum_{v \in \mathcal S} c_v^\dagger\ket{0}_{\mathcal S}$
	is the W-state over the qubits $\mathcal S$ (an equal superposition over Hamming weight 1 strings).

	The protocol is then as follows.
	We first use (fast) \swap{}s between each qubit and its ancilla 
	so all data qubits in the graph are in the state $\ket{0}$.
	We now pick some vertex $v \in V(G_L)$ and show how to route the state originally at $v$ to $\sigma(v)$.
	We \swap{} the data qubit at $v$ with its ancilla to return $v$ to its initial state.
	Then we evolve by the Hamiltonian $W(v, V(G_L) \setminus \set{v})$ to encode the state on $v$
	on the data qubits associated with $V(G_L) \setminus \set{v}$,
	creating a state similar to \cref{eq:wstateencoding},
	followed by the inverse operation $W(v_c, V(G_L) \setminus \set{v})^\dagger$.
	Overall, this sends the state from $v$ to the central vertex $v_c$ in total time $2T = \pi/\sqrt{n-1}$
	(see also \cref{fig:state-transfer}).
	We repeat this process to transfer the qubit from $v_c$ to $\sigma(v)$ in time $2T$.
	Then we \swap{} the qubit at $\sigma(v)$ with its ancilla.
	If the qubit that is now at $\sigma(v)$ needs to be routed, we follow an analogous procedure and send it to $\sigma(\sigma(v))$.
	If it does not, we pick some other vertex in $V(G_R)$ that still needs to be routed.
	We iterate in this way, alternately handling a vertex from $G_L$, then $G_R$,
	until all vertices are routed to their destination ancillas.
	Finally, we simultaneously \swap{} all qubits with their ancillas to finish the routing.
	The total time is $4T\cdot 2n = \bigo{\sqrt{n}}$.
\end{proof}

We can now generalize the algorithm to all permutations on $B_{2n}$.
\begin{corollary}\label{cor:barbellRouting}
	For the vertex barbell graph $B_{2n}$,
	we have
	\begin{equation}
		\hrAnumber{B_{2n}} = \bigo{\sqrt{n}}.
	\end{equation}
\end{corollary}
\begin{proof}
	Let $\sigma$ be any permutation of the vertices $V(G)$.
	First, we \swap{} $\sigma(v_c)$ with its ancilla, \swap{} $v_c$ with $\sigma(v_c)$,
	and finally \swap{} $\sigma(v_c)$ with its ancilla again.
	Then, we route all vertices that are permuted only within $G_L$ or $G_R$
	in $\bigo{1}$ time using \swap{}s since $\rnumber{K_n} \le 2$~\cite{Alon1994}.
	Consider now the vertex barbell subgraph of the remaining vertices that need to move
	between $G_L$ and $G_R$, together with $v_c$.
	This routing can be done in time $\bigo{\sqrt{n}}$ by \cref{thm:barbellSwap}
	and starting at $\sigma(v_c)$.
	Finally, we \swap{} $\sigma^{-1}(v_c)$ with its ancilla, and then with $v_c$.
\end{proof}

\Cref{cor:barbellRouting} shows a quadratic separation
\begin{equation}
	\qrt_{\mathrm{a}}(B_{2n}) = \bigomega{\hrAnumber{B_{2n}}^2}.
\end{equation}
It also shows $\hrAnumber{B_{2n}} \notin \bigomega*{c(B_{2n})^{-1}} = \bigomega{n}$, so \cref{cor:qrtVertexLB} does not generalize to Hamiltonian routing with ancillas.

\subsection{Optimal routing with fast local interactions}\label{sec:optimalRouting}
In this section, we show optimal routing
for stronger models of classical and Hamiltonian routing
that allow arbitrarily fast interactions within partitions of the graph $G$.
Given a partition $X \subsetneq V(G)$,
let us define the \emph{fast classical routing model}
as classical routing with arbitrarily fast \swap{} operations within the vertex-induced subgraphs $G[X]$ and $G[\bar X]$.
Then we can define the $X$-\emph{fast routing number} of $G$, $\rfnumber{G, X}$,
as the worst-case time to route
any permutation in the fast classical routing model for a given graph $G$ and partition $X$.
We denote with subscript ``f'' that the interactions within the partitions $X$ and $\bar X$ are fast.

We now show that $\rfnumber{G, X} = \ceil{\frac{\abs{X}}{\abs{\matching{\partial X}}}}$
for any connected simple graph $G$ and partition $X \subseteq V(G)$ with $\abs{X} \le \abs{V(G)}/2$.
The upper bound is given by the following routing algorithm.
\begin{theorem}\label{thm:fastClassicalRouting}
	For any connected simple graph $G$ and partition $X \subseteq V(G)$ with $\abs{X} \le \abs{V(G)}/2$,
	\begin{equation}
		\rfnumber{G, X} \le \ceil*{\frac{\abs{X}}{\abs{\matching{\partial X}}}}.
	\end{equation}
\end{theorem}
\begin{proof}
	For any permutation $\pi$ to be routed,
	call the $k$ vertices $x \in X$ such that $\pi(x) \notin X$ \emph{marked}.
	Similarly, we mark the vertices $x' \in \bar X$ such that $\pi(x') \notin \bar X$.
	Fix a maximum matching $\matching{\partial X} = \set{(x_i, x_i')}_{i=1}^k$ with $x_i \in X$ and $x_i' \in \bar X$.
	We repeat the following two steps for $\ceil{\abs{X}/\abs{\matching{\partial X}}}$ times:
	\begin{enumerate}
		\item Route as many marked vertices in $X$ as possible to $x_1,\dots,x_k$ in order,
			and route as many marked vertices in $\bar X$ as possible to $x_1',\dots,x_k'$ in order.
		\item\label{item:fastClassical2} Perform parallel \swap{}s along
			$(x_i,x_i')$ for all $i$ less than the number of remaining marked vertices.
	\end{enumerate}
	This routes all marked vertices to their destination partitions.
	The only contribution to the $X$-fast routing number of $G$ is a unit contribution of \cref{item:fastClassical2}
	every iteration.
	Finally, we route all qubits within $X$ and $\bar X$ to their destinations using fast \swap{}s.
\end{proof}

The lower bound \cref{eq:qrtExpansionBoundSimple} on the gate-based quantum routing number
applies to the model with fast interactions
since STE still upper bounds the change in entropy for any unitary acting on $\matching{\partial X}$.
% (The lower bound \cref{eq:qrtExpansionBound} does not apply since fast interactions are possible between $\delta X$ and $\bar X \setminus \delta X$.)
The lower bound (rounded up, since the routing number cannot be a fraction)
is attained by the $X$-fast classical routing algorithm for $G$,
and thus the algorithm is optimal for all gate-based models.

Given a partition $X \subsetneq V(G)$,
let us define the \emph{fast Hamiltonian routing model}
as Hamiltonian routing with ancilla and arbitrary interactions within the vertex-induced subgraphs $G[X]$ and $G[\bar X]$.
Then we can define the $X$-\emph{fast Hamiltonian routing time} of $G$, $\hrfanumber{G, X}$,
as the worst-case time to route
any permutation in the fast Hamiltonian routing model for a given graph $G$ and partition $X$.
The additional subscript ``a'' indicates the presence of an ancilla at every qubit.

We give an $X$-fast Hamiltonian routing algorithm that attains the lower bound on Hamiltonian routing
up to a constant factor $3\pi\sieConst/8$ (where $\sieConst$ is the constant of SIE, \cref{lem:sie})
for any graph $G$.
We first prove that it is possible to perform any two-qubit unitary
in the fast Hamiltonian routing model in time $\bigo{1/\abs{\partial X}}$.
\begin{lemma}\label{lem:unitaryFastHamiltonian}
	For a connected simple graph $G$,
	any two-qubit unitary $U$ can be performed in the $X$-fast Hamiltonian routing model
	with partition $X \subsetneq V(G)$ in time at most $\frac{1}{\abs{\partial X}}$.
\end{lemma}
\begin{proof}
	We show how to perform a \cz{} operation between $v, u \in V(G)$ in time $t = \frac{1}{3\abs{\partial X}}$.
	By a decomposition of $U$ into at most 3 \cz{} operations plus single-qubit operations~\cite{Vatan2004}, the result follows.
	The result is trivial if $v$ and $u$ are both within $X$ or $\bar X$.

	Suppose, without loss of generality, $v \in X$ and $u \in \bar X$.
	We use (fast) \swap{}s between qubits on the boundary $\delta X \cup \delta \bar X \setminus \set{u,v}$ and their ancillas.
	Suppose the qubit at $v$ is in the state $a_0\ket{0} + a_1\ket{1}$
	and $u$ is in the state $a_0'\ket{0} + a_1'\ket{1}$
	(by linearity, the protocol also works if $v$ and $u$ are initially entangled with other qubits).
	We then encode the state of $v$ onto $\delta \bar X$ as
	\begin{equation}
		a_0\ket{0\dots0} + a_1\ket{1\dots1} = a_0\ket{\bar 0} + a_1\ket{\bar 1}
	\end{equation}
	by fast unitaries.
	Similarly, with some abuse of notation, we encode the state of $u$
	onto $\delta X$ as $a_0'\ket{\bar 0}+a_1'\ket{\bar 1}$
	where we disregard the different register sizes with the overline notation.

	Notice that a Hamiltonian on $x,x' \in V(G)$
	\begin{equation}
		\frac{3\pi}{4}\paren{\idm - Z_x - Z_{x'} + Z_xZ_{x'}} = 3\pi\ket{11}\bra{11}_{xx'},
	\end{equation}
	with $Z_x$ a Pauli-Z operator acting on the qubit at $x$,
	consists of local terms and one normalized $ZZ$ interaction.
	We can therefore evolve by the Hamiltonian
	\begin{equation}
		H = 3\pi \sum_{xx' \in \partial X}\idm_{\overline{xx'}} \otimes \ket{11}\bra{11}_{xx'},
	\end{equation}
	where $\idm_{\overline{xx'}}$ is the identity operator on all subsystems besides $xx'$,
	using fast local unitaries and $ZZ$ interactions along each edge in $\partial X$.
	By commutativity of the terms in $H$, we see that
	\begin{multline}
		e^{-itH} (a_0 \ket{\bar 0} + a_1 \ket{\bar 1}) \otimes (a_0'\ket{\bar 0} + a_1' \ket{\bar 1}) =\\
		a_0a_0' \ket{\bar0 \bar0} + a_0a_1'\ket{\bar0 \bar1} + a_0'a_1\ket{\bar1\bar0} + a_1a_1' e^{-i3\pi t\abs{\partial X}} \ket{\bar1 \bar1}.
	\end{multline}
	After performing our initial \swap{} and encoding operations in reverse,
	we have applied a \cz{} operation between $u$ and $v$ in time $t$.
\end{proof}

With the ability to quickly perform arbitrary two-qubit operations in the fast Hamiltonian routing model,
we give a routing algorithm with the $X$-fast Hamiltonian routing time of $G$ upper bounded by $\abs{X}/\abs{\partial X}$.
\begin{theorem}\label{thm:fastHamiltonianRouting}
	For a connected simple graph $G$ and partition $X \subsetneq V(G)$,
    \begin{equation}
    	\hrfanumber{G, X} \le \frac{\abs{X}}{\abs{\partial X}}.
    \end{equation}
\end{theorem}
\begin{proof}
	Let $\pi$ be any permutation to be routed,
	and suppose there are $k$ \emph{marked} vertices $x \in X$ such that $\pi(x) \notin X$.
	Then there are also $k$ marked vertices $x' \in \bar X$ such that $\pi(x') \notin \bar X$.
	We perform $k$ \swap{}s between each pair of marked vertices,
	which, by \cref{lem:unitaryFastHamiltonian}, can be done in time $\frac{k}{\abs{\partial X}}$.
	Finally, we use fast local \swap{}s to route all qubits to their destination.
	The result follows since $k \le \abs{X}$.
\end{proof}

\Cref{thm:entanglementCapacityCut} also bounds the entanglement capacity of any Hamiltonian acting
on the edge boundary $\partial X$ in the fast Hamiltonian model.
Therefore, \cref{eq:hrnumberLBX} implies
\begin{equation}
    \hrfanumber{G, X} \ge \frac{8}{3\pi\sieConst} \frac{\abs{X}}{\abs{\partial X}}\label{eq:fhqrLB}
\end{equation}
for any $X \subsetneq V(G)$ with $\abs{X} \le \abs{V(G)}/2$.
It follows that \cref{thm:fastHamiltonianRouting} is tight up to a multiplicative constant $3\pi\sieConst/8 \le 4.713$.

When we compare the $X$-fast Hamiltonian routing time
with the $X$-fast classical routing time (even with ancilla),
we see that
\begin{equation}
	\frac{\rfnumber{G,X}}{\hrfanumber{G,X}} = \bigtheta*{\frac{\abs{\partial X}}{\abs{\matching{\partial X}}}}.
\end{equation}
The vertex barbell graph with the partition $V(G_L)$ is an example
where a speedup of $\bigtheta{n}$ is realized.

\section{Conclusion}
In this paper, we have explored the power of gate-based and Hamiltonian models of quantum routing,
investigating both lower bounds and separations.
We showed conditions on the spectrum of the architecture graph for a superpolynomial separation.
In particular, our conditions exclude bounded-degree graphs from exhibiting a superpolynomial separation.
We also gave an example graph~\cite{Raghavan1994} where diameter-based lower bounds and known classical routing algorithms~\cite{Alon1994} cannot exclude such a separation.

One natural open question is whether the star graph $S_n$, which has $\lambda(S_n) = 1$,
can exhibit a superconstant quantum routing separation.
While our results imply that gate-based quantum routing essentially gives no improvement over classical routing since $\qrnumber{S_n} \ge n - 1$,
the same cannot be said for Hamiltonian routing, for which the corresponding lower bound is trivial.
In fact, if the Hamiltonian model is strengthened by allowing a constant number of ancillas per qubit,
a quadratic separation holds on the vertex barbell graph,
which also exhibits a similar vertex bottleneck.
By allowing fast interactions within certain regions of the graph,
we can give optimal routing algorithms for gate-based and Hamiltonian models
and exhibit a speedup from $\bigtheta{n}$ for gate-based models
to $\bigo{1}$ for Hamiltonian models.

Our depth (or time) lower bounds can be strengthened to include computational models with local operations and classical communication (LOCC).
LOCC models give a stronger class of quantum routing
and would allow, e.g., teleportation to bridge long distances.
Trivially, this can exceed the Lieb-Robinson velocity~\cite{Lieb1972}
and seemingly invalidates simple lower bounds based on the diameter of the graph.
\citeauthor{Piroli2021}~\cite{Piroli2021} showed LOCC circuit lower bounds on state preparation for lattices
and inspired us to show similar state preparation results for general interaction graphs
and to lower bound routing.
Since our depth (and time) lower bounds follow from entropic arguments
and the entropy is non-increasing under LOCC,
we see that STE, SIE, and our state preparation bounds (\cref{lem:gateBasedStatePreparation,cor:hamiltonianStatePreparation}) generalize to models including LOCC when the entropy is non-decreasing.
Thus our quantum routing bounds (\cref{thm:qrtExpansionBound,thm:hqrtEdgeBound}) 
also generalize to models including LOCC.
How much stronger models of routing with LOCC can be is studied in~\cite{loccRouting}.

\subsection*{Acknowledgements}
We thank Minh Tran for suggesting the possibility of using the state transfer protocol of \cite{Guo2019},
Michael Gullans for pointing out properties of Erd\"{o}s-R\'{e}nyi graphs,
and Dhruv Devulapalli and Andrew Guo for helpful discussions.

A.B.\ and A.V.G.\ acknowledge funding by the NSF PFCQC program, ARO MURI, DoE QSA,  DoE ASCR Quantum Testbed Pathfinder program (award No.~DE-SC0019040), NSF QLCI (award No.~OMA-2120757), DoE ASCR Accelerated Research in Quantum Computing program (award No.~DE-SC0020312), DARPA SAVaNT ADVENT, AFOSR, AFOSR MURI, and U.S.~Department of Energy Award No.~DE-SC0019449.
E.S.\ and A.M.C.\ acknowledge support by the U.S.\ Department of Energy, Office of Science, Office of Advanced Scientific Computing Research, Quantum Testbed Pathfinder program (award number DE-SC0019040) and the U.S.\ Army Research Office (MURI award number W911NF-16-1-0349).
E.S.\ acknowledges support from an IBM PhD Fellowship.

\printbibliography%

\appendix
\section{Asymptotic equivalence  of the matching expansion and vertex expansion}\label{app:matchingExpansion}
We show that the matching expansion~\eqref{eq:matchingExpansion} is equivalent to the vertex expansion,
i.e., $c(G) = \bigtheta{\matching{G}}$.
The lower bound $\matching{G} \le c(G)$ follows from the trivial bound $\abs{\matching{\partial X}} \le \abs{\delta X}$ for any $X$.
The following theorem provides the upper bound.
\begin{theorem}\label{thm:vertexCutMatchingBound}
	For any simple graph $G$,
	\begin{equation}
		c(G) \le 2\matching{G} + \bigo{\matching{G}^2}.
	\end{equation}
\end{theorem}
\begin{proof}
    We note that $m(G) \in [0,1]$.
    If $m(G) = 1$, then the theorem holds since $c(G) \in [0,1]$.
    
    We now consider the case $m(G) \in [0,1)$.
	Let $X \subseteq V(G)$ be a partition that attains the minimum in the matching expansion,
	and $Y \subseteq V(G)$ is the set of vertices in $\matching{\partial X}$.
	The set $X' \coloneqq X \setminus Y$ is non-empty because $\matching{G} < 1$.
	We show
	\begin{equation}\label{eq:matchingToVertex}
		\abs{\delta X'} \le \abs{Y} = 2\abs{\matching{\partial X}}.
	\end{equation}
	Suppose, toward a contradiction,
	that there are adjacent vertices $x' \in X'$ and $x \in \overline{X'} \setminus Y$.
	Then $\matching{\partial X}$ is not maximal since $\matching{\partial X} \cup \set{(x,x')}$ is a larger matching.
	Therefore, $\delta X'$ must only consist of vertices in $Y$,
	giving $\abs{\delta X'} \le \abs{Y}$ as claimed.
	
	It follows from \cref{eq:matchingToVertex} that
	\begin{equation*}
		c(G) \le \frac{\abs{\delta X'}}{\abs{X'}} \le \frac{2\abs{\matching{\partial X}}}{\abs{X'}} = \frac{2\abs{\matching{\partial X}}}{\abs{X} -\abs{\matching{\partial X}}} = \frac{2\matching{G}}{1 - \matching{G}}.\qedhere
	\end{equation*}
\end{proof}
\end{document}